\documentclass[12pt,twoside]{article}
\usepackage[dvips]{graphicx}
\usepackage{float}
\usepackage{latexsym}
\usepackage{amsmath,mathrsfs,bm,amssymb,color}
\catcode`\@=11
\def\newsymbol#1#2#3#4#5{\let\next@\relax%
 \ifnum#2=\@ne\else%
 \ifnum#2=\tw@\let\next@\msyfam@\fi\fi%
 \mathchardef#1="#3\next@#4#5}
\def\mathhexbox@#1#2#3{\relax%
 \ifmmode\mathpalette{} {\m@th\mnnathchar"#1#2#3}
 \else\leavevmode\hbox{$\m@th\mathchar"#1#2#3$}\fi}
\font\tenmsy=msbm10
\font\sevenmsy=msbm7
\font\fivemsy=msbm5
\newfam\msyfam
\textfont\msyfam=\tenmsy
\scriptfont\msyfam=\sevenmsy
\scriptscriptfont\msyfam=\fivemsy
\edef\msyfam@{\hexnumber@\msyfam}
\def\Bbb#1{\fam\msyfam\relax#1}
\catcode`\@=\active

\newtheorem{theorem}{Theorem}[section]
\newtheorem{proposition}[theorem]{Proposition}
\newtheorem{lemma}[theorem]{Lemma}
\newtheorem{corollary}[theorem]{Corollary}
\newtheorem{definition}[theorem]{Definition}

\newtheorem{remark}[theorem]{Remark}

\newtheorem{assumption}[theorem]{Assumption}

\newcommand{\w}{\sqrt\omega}
\newcommand{\ww}{\frac{1}{\sqrt\omega}}
\newcommand{\A}{\mathscr{A}}
\newcommand{\B}{\mathscr{B}}
\newcommand{\W}{\mathcal {W}}
\newcommand{\E}[1]{[#1]_{\ex}}
\newcommand{\J}{J_P^\ex}

\newcommand{\vA}{\vec {\mathscr{A}}}
\newcommand{\pp}{{\rm TL}}
\newcommand{\oi}{{{\rm ex}}} 
\newcommand{\skima}{\!\!\frac{}{}}
\newcommand{\eq}[1]{\begin{equation}\label{#1}}
\newcommand{\en}{\end{equation}}
\newcommand{\eqn}{\begin{eqnarray*}}
\newcommand{\enn}{\end{eqnarray*}}
\newcommand{\eqnn}{\begin{eqnarray}}
\newcommand{\ennn}{\end{eqnarray}}
\newcommand{\proof}{{\noindent \it Proof:\ }}
\newcommand{\qed}{\hfill {\bf qed}\par\medskip}
\newcommand{\BR}{{{\Bbb R}^3 }}
\newcommand{\bi}{\begin{description}}
\newcommand{\ei}{\end{description} }
\newcommand{\CC}{{{\Bbb C}}}
\newcommand{\RR}{{\Bbb R}}
\newcommand{\bl}[1]{\begin{lemma}\label{#1}}
\newcommand{\el}{\end{lemma}}
\newcommand{\bc}[1]{\begin{corollary}\label{#1}}
\newcommand{\ec}{\end{corollary}}
\newcommand{\bt}[1]{\begin{theorem}\label{#1}}
\newcommand{\et}{\end{theorem}}
\newcommand{\bp}[1]{\begin{proposition}\label{#1}}
\newcommand{\ep}{\end{proposition}}
\newcommand{\br}[1]{\begin{remark}\label{#1}}
\newcommand{\er}{\end{remark}}
\newcommand{\im}[2]{(#1|#2)}
\newcommand{\mass}{m_{\rm eff}}

\newcommand{\V}{\mathscr{V}}
\newcommand{\kak}[1]{(\ref{#1})}
\newcommand{\LR}{{L^2(\BR)}}

\newcommand{\fff}{{\ffff_{\rm b}}}
\newcommand{\ffff}{\mathscr{F}}
\newcommand{\fffff}{\ffff_{\rm fin}}

\newcommand{\is}{\inf{\rm Sp}}
\newcommand{\f}{^{-1}}
\newcommand{\lk}{\left(}
\newcommand{\rk}{\right)}
\newcommand{\lkk}{\left\{}
\newcommand{\rkk}{\right\}}

\newcommand{\T}{{\rm T}}
\newcommand{\+}{~[\dot+]~}
\newcommand{\ex}{{\rm ex}}

\newcommand{\add}{a^{\ast}}
\newcommand{\PI}{\prod_{i=1}^n}

\newcommand{\addd}{a^{\dagger}}

\newcommand{\bdd}{b_P^{\ast}}
\newcommand{\df}{\mathscr{D}_{\rm free}(\mathcal {K})}
\newcommand{\dff}{\mathscr{D}_{\rm free}(\hhh)}

\newcommand{\ass}{a^\sharp}
\newcommand{\bss}{b^{ \sharp}}
\newcommand{\ov}[1]{\overline{#1}}
\newcommand{\hf}{H_{\rm f}}
\newcommand{\hff}{H_{\rm f}^{\pp}}

\newcommand{\half}{\frac{1}{2}}
\newcommand{\han}{{1/2}}

\newcommand{\vvv}[1]
{\left[\!\!\!\begin{array}{c}#1\end{array}\!\!\!\right]}
\newcommand{\MM}[4]
{\left[\!\!\!\begin{array}{cc}#1&#2\\ #3&#4\end{array}\!\!\!\right]}

\newcommand{\hhh}{{\mathscr{H}}}

\newcommand{\ar}[2]{\begin{array}{#1} #2 \end{array}}
\newcommand{\mmm}{\sum_{j=1}^3}
\newcommand{\jjj}{\sum_{j=1}^{3}}

\newcommand{\mmml}{\sum_{l=1}^3}
\newcommand{\ab}[1]{(H\Psi | #1 \Phi)}
\newcommand{\iii}{\sum_{i=1}^3}
\renewcommand{\d}{\displaystyle}
\newcommand{\vp}{{\hat \varphi}}
\newcommand{\non}{\nonumber}

\makeatletter
\@addtoreset{equation}{section}
\makeatother

\title
{\sc Physical state for non-relativistic quantum electrodynamics}

\author{
Fumio Hiroshima\thanks{e-mail: hiroshima@ math.kyushu-u.ac.jp}\\
Department of Mathematics\\
Kyushu University
\\
\\
 Akito Suzuki\thanks{e-mail: akito@math.sci.hokudai.ac.jp}
 \\
 Department of Mathematics\\
  Hokkaido University
 }
\date{\today}
\begin{document}
\setlength{\baselineskip}{18pt}
\maketitle
\begin{abstract}
A physical subspace and physical Hilbert space
associated with
asymptotic fields of
nonrelativistic quantum electrodynamics
are constructed through
the Gupta-Bleuler procedure.
Asymptotic completeness is shown and a
physical Hamiltonian is defined on the
physical Hilbert space.
\end{abstract}

\section{Introduction}
\subsection{The Gupta-Bleuler formalism}
Quantization of the electromagnetic field
does not cohere with normal postulates such as
Lorentz covariance and existence of a positive definite metric on some Hilbert
space. This means that
we chose to quantize in a manner sacrificing manifest Lorentz
covariance;
conversely if the electromagnetic field is quantized in a manifestly
covariant fashion, the notion of a positive
definite metric must be sacrificed and
the existence of negative probability arising from
the indefinite metric
renders invalid
a probabilistic interpretation of quantum field theory.
  One prescription for quantization of the electromagnetic field
   in a Lorentz covariant manner
  is the Gupta-Bleuler procedure \cite{Bleuler, Gupta}.
This procedure provides a covariant procedure for quantization at the cost of a cogent physical interpretation.

In this paper  we will consider  the so-called nonrelativistic
quantum electrodynamics (NRQED).
A significant point is that NRQED is nonrelativistic with respect to the motion of an electron only;
the electromagnetic field is always relativistic.
Although it is customary to adopt the Coulomb gauge
in the theory of NRQED,
 it   can be
 investigated using the Lorentz gauge by the Gupta-Bleuler approach \cite{Ba1}. This will be rigorously pursued in this paper.

{\it Indefinite metric and Lorentz condition:}
Let $\A_\mu=\A_\mu(x,t)$, $\mu=0,1,2,3$,
be a quantized radiation field
and $\dot \A_\mu=\dot \A_\mu(x,t)$ its time derivative.
$\A_\mu$ and its time derivative $\dot \A_\nu$
satisfy the 
commutation relations
	\begin{align}
	\label{covccr1}
		& [\A_\mu(x,t), \dot \A_\nu(x',t)]=-ig_{\mu\nu}\delta(x-x'), \\
	\label{covccr2}
		& [\A_\mu(x,t),\A_\nu(x',t)]=0,\\
&\label{covccr3}[\dot \A_\mu(x,t), \dot \A_\nu(x',t)]=0,
	\end{align}
where $g_{\mu\nu}$ is the metric tensor given by
\kak{mink}.
An inevitable consequence of the commutation relations \eqref{covccr1}-\eqref{covccr3}
is to introduce an indefinite metric $(\cdot | \cdot)$
onto the state space.
This creates  problems in physical interpratations and in formulating
things in a mathematically well defined way. For example,
the Hamiltonian $H$ cannot be defined as a self-adjoint operator
and so
the time-evolution $e^{itH}$ is not unitary.
So we have to investigate
questions concerning the domain of  $e^{itH}$.

In addition to the indefinite metric, the Lorentz condition also
poses a dilemma.
We impose the Lorentz condition:
\eq{sa7}
\partial^\mu\A_\mu(x,t)=0
\en
as an operator identity.
Here and in what follows
$\partial^\mu X_\mu$ is the conventional abbreviation for
$\partial^\mu X_\mu=\partial_t X_0
-\partial_{x^1}X_1-\partial_{x^2}X_2-\partial_{x^3}X_3$.
Under \kak{sa7}, as is well known, we find that
the conventional Lagrangian formalism is not available.

To resolve this difficulty, in the Gupta-Bleuler procedure mentioned below,
we first single out the so-called physical subspace from
the Lorentz condition,
and it is
 required that the Lorentz condition
 is valid only in terms of expectation values on the physical subspace.
The sesquilinear form $( \cdot|\cdot)$ restricted to the physical subspace is merely semidefinite.
So, we define the physical Hilbert space to be the quotient space
of the physical subspace divided by the subspace with zero norm with respect to $( \cdot|\cdot)$.
This space has a positive definite form,
and
a  self-adjoint Hamiltonian can also be derived.

{\it Gupta-Bleuler formalism:} Here we present
an outline of the Gupta-Bleuler formalism in NRQED for the reader's convenience,
without mathematical rigor.
Let $H_{\rm tot}$ be the full Hamiltonian for NRQED with form factor
$\varphi$ .
Note that $H_{\rm tot}$ is not self-adjoint.
Let $X$ be an operator.
Generally, a solutions $X(t)$ to the Heisenberg equation:
\eq{A(t)eq}
		\frac{d}{dt}X(t) = i[H_{\rm tot}, X(t)], \quad X(0) = X,
	\end{equation}
is called a
Heisenberg operators of $X$ associated with
$H_{\rm tot}$.
Since $H_{\rm tot}$ is, however,  not self-adjoint,
intuitively a solution to \kak{A(t)eq} is possibly
not unique.
In order to ensure uniqueness
we give an alternative definition of  Heisenberg operators
in Definition~\ref{DHO+}.

Let $p$ and $q$ be the momentum
and position operators respectively
of an electron, and
$\A(f)=(\A_0(f),\vA(f))$
the smeared electromagnetic field, i.e.,
$$\A(f)=\int f(x) \A(x) dx.$$
Let $p(t)$, $q(t)$ and
 $\A(f,t)=(\A_0 (f,t),
\vA(f,t))$
be the Heisenberg operators of $p$,  $q$
and  $\A(f)$, respectively.
We denote $\A(f,t)$ by
$$\A(f,t)=\int f(x) \A(x,t) dx.$$
Then formally the equations
	\begin{align}
	\label{Maxwell1}
		\square \vA(x,t) & = {\vec J}(x,t), \\
	\label{Maxwell2}
		\square {\A }_0(x,t) & = \rho(x,t)
	\end{align}
can be derived. Here
$\rho$ and ${\vec J}$ are the charge and the current density
of the electron, respectively,  given by
\begin{eqnarray}
&&
\label{ka1}
 \rho(x,t)=e\varphi(x-q(t)),\\
 &&
\label{ka2}
  \vec J(x,t)=\frac{e}{2}\left( \skima \varphi(x-q(t))\vec v(t) + \vec v(t) \varphi(x-q(t))\right)
  \end{eqnarray}
  where
$e$ denotes the charge on an electron and $\vec v(t)$ the velocity:
  $$\vec v(t) = \dot q(t)=\frac{1}{m}
  \lk
  p(t)-e\int
  \vA(z,t) \varphi(z-q(t)) dz\rk.
  $$
It can be seen  from \kak{ka1} and \kak{ka2} that the 4-current
$j=(\rho, \vec J) = (j^0,j^1,j^2,j^3)$
satisfies the continuity equation
\eq{con}
 \partial^\mu j_\mu =0.
 \en
By this, together with \eqref{Maxwell1} and \eqref{Maxwell2},
the kernel $\A(x,t)$ automatically satisfies that the condition
	\begin{equation}
	\label{freeAmu}
		\square \partial^{\mu}\A _{\mu}(x,t) = 0.
	\end{equation}
Equation \eqref{freeAmu} tells us that
$\partial^{\mu}\A _{\mu}$ is a free field and hence
formally, it can be described  in terms of
some
  annihilation operator $c(k)$
and the creation operator $c^\dagger(k)$ by
\eq{new}
 \partial^{\mu}\A _{\mu}(x,t)
		= \int \left( c(k)e^{-i|k|t+ikx}
			+ c^{\dagger}(k)e^{i|k|t-ikx} \right)dk. \en
The term including  the factor
$e^{-i|k|t}$ in \kak{new}
is called the positive frequency part
of $\partial^{\mu}\A _{\mu}(x,t)$
and written as $[\partial^{\mu}\A _{\mu}]^{(+)}(x,t)$.
On the other hand the negative frequency part $[\partial^\mu\A_\mu]^{(-)}(x,t)$ is defined
by the term including $e^{+i|k|t}$.
As is mentioned above
the Lorentz condition \eqref{sa7} is not valid as an operator identity.
We may demand that some state $\Psi$ should satisfy
$\partial^\mu\A_\mu(x,t)\Psi=0$.
This is however too severe a condition to demand,
since
$[\partial^{\mu}\A _{\mu}]^{(+)}(x,t)
\Psi+
[\partial^{\mu}\A _{\mu}]^{(-)}(x,t)\Psi=0$ and
the negative frequency part contains creation operators, so not even
the vacuum could satisfy this identity.
However, since the positive frequency part contains the annihilation operator,
we could adopt the less demanding requirement
	\begin{equation}
	\label{GB}
		[\partial^{\mu}\A _{\mu}]^{(+)}(x,t)\Psi = 0.
	\end{equation}
The state $\Psi$ in \kak{GB}
 is called the physical state,
  and
\kak{GB}
is called the Gupta-Bleuler subsidiary condition.
The set of physical states is denoted by
  $\mathscr{V}_{\rm phys}$ and is called
    the physical subspace.
Moreover the Lorentz condition is realized as the expectation value on the physical subspace:
	\[ ( \Psi | \partial^\mu A_\mu \Psi) = 0, \quad \Psi \in \mathscr{V}_{\rm phys}. \]

In much of the physical literature little attention
is paid to the existence of a nontrivial physical subspace.
The absence of a physical subspace
was, however,  recently pointed out
in \cite{S1}.
In this paper we want to derive sufficient conditions
for the existence of a physical subspace and
characterize such a subspace in the NRQED framework.

\subsection{Main results and plan of the paper}
Our main concern in this paper is to develop the Gupta-Bleuler formalism for NRQED, and
to characterize the physical subspace
rigorously.
   The physical subspace, however, can be trivial because of the infrared singularity \cite{S1}.
The difficulty in construction of the physical subspace
of our system is due to the fact that
$H_{\rm tot}$ is $\eta$-self-adjoint but not self-adjoint on a Klein space.
Therefore one cannot realize the solution of the Heisenberg equation \kak{A(t)eq}
as $e^{itH_{\rm tot}}\A e^{-itH_{\rm tot}}$.

In this paper we introduce a
dipole approximation to $H_{\rm tot}$
 to reduce this difficulty.
Let $H$ denote the Hamiltonian with
a dipole approximation.
Even so, although the Hamiltonian $H$ is
$\eta$-self-adjoint, it is not yet self-adjoint.
However, thanks to the dipole approximation,
we can construct the Heisenberg operators
 $\A_0 (f,t)$ and $\vA(f,t)$ exactly.
See Theorem \ref{h-o}.

On the other hand, there is a disadvantage in using
the dipole approximation.
Unfortunately, with this approximation, the system does
not conserve the $4$-current
$j_{\rm dip}=(\rho_{\rm dip}, \vec J_{\rm dip})$.
In fact, the $4$-current in
the dipole approximation turns out to be
\begin{eqnarray}
&&
\label{ka11}
 \rho_{\rm dip}(x,t)=e\varphi(x),\\
 &&
\label{ka21}
  \vec J_{\rm dip}(x,t)=\frac{e}{2}
  \left( \skima \varphi(x)\vec v_{\rm dip}
  (t) + \vec v_{\rm dip}(t) \varphi(x)\right)
  \end{eqnarray}
  where
  $$\vec v_{\rm dip}
  (t) = \frac{1}{m}
  \lk
  p(t)-e\int
  \vA(z,t) \varphi(z) dz\rk
  $$
and
\eq{ka3}
\partial^\mu {j_{\rm dip}}_\mu\not=0.
\en
Hence $\partial^\mu \A_\mu$ is not a free field in the sense of
\kak{freeAmu}, and
we lose the method of
defining
the positive frequency part $[\partial^{\mu}\A _{\mu}]^{(+)}$.
Therefore, in the dipole approximation, the physical subspace
cannot be defined in the usual way.

Nevertheless the asymptotic field
provides a tool for employing  the Gupta-Bleuler formalism.
Decompose $H$ with respect to the spectrum of the electron momentum:
\eq{spdec}
H=\int_{\BR}^\oplus H_P dP.
\en
We shall generally consider $H_P$ for an arbitrary fixed $P\in\BR$ throughout this paper.

The main results of this paper are
\begin{description}
\item[(i)] Asymptotic completeness of $H_P$ based on the LSZ method (Theorem \ref{asy});
\item[(ii)] Characterization of the physical subspace (Theorems \ref{suzuki9} and \ref{suzu34}) and the physical Hilbert space \kak{phi};
\item[(iii)] Construction of the physical scattering operator (Theorem \ref{suzu2});
\item[(iv)] Construction of the physical self-adjoint Hamiltonian (Theorem \ref{suzu1}).
\end{description}

{\bf (i)}
The explicit form of the Heisenberg operator with respect to $H_P$ allows us to
construct the asymptotic fields
$\A _{\mu}^{\rm out/in}(f,t,P) $ exactly
   and we prove the asymptotic completeness
   in    Theorem \ref{asy}.
As far as we know,
the  asymptotic completeness of NRQED with the dipole approximation
was proven initially by Arai \cite{a83,a83b} but for the model
without both a scalar and a longitudinal component.
See also Spohn \cite{sup97}.   We extend this to our case.

{\bf (ii)}
$\A _{\mu}^{\rm out/in}(f,t,P) $ is a free field defined in terms of asymptotic annihilation and creation operators, therefore so is $\partial^{\mu}\A _{\mu}^{\rm out/in} (f,t,P)$.
Therefore one can define the non-trivial physical subspace
$\mathscr{V}_{P,{\rm phys}}^{\rm out/in} $
associated  with  $\partial^{\mu}\A _{\mu}^{\rm out/in}(f,t,P) $ by the Gupta-Bleuler subsidiary condition
\eq{asgb}
 [\partial^{\mu}\A _{\mu}^{\rm out/in} ]^{(+)}(f,t,P)\Psi = 0.
 \en
We characterize $\V_{P,{\rm phys}}^{\rm out/in}$
and prove that $\V_{P,{\rm phys}}^{\rm out/in}$ is positive semi-definite
in Theorem \ref{suzuki9}.
Moreover,  in Theorem \ref{suzu34},
we show that $${\mathscr{V}_{P,{\rm phys}}^{\rm out}}\not =
{\mathscr{V}_{P,{\rm phys}}^{{\rm in}}}$$
which cannot occur in the case where the 4-current is conserved \cite{Sunakawa, IZ}.
The physical subspace
is decomposed as the direct sum:
$\mathscr{V}_{P,{\rm phys}}^{\rm out/in}
=\mathscr{V}_P^{\rm out/in}[ \dot +]
\mathscr{V}_{P,{\rm null}}^{\rm out/in}$,
where $\mathscr{V}_{P,{\rm null}}^{\rm out/in}$ is
the null space with respect to an indefinite metric
and the Hamiltonian leaves it  invariant.
Then the physical Hilbert space
is given as the quotient space
$$
\hhh _{P,{\rm phys}}^{\rm out/in}=
\mathscr{V}_{P,{\rm phys}}^{\rm out/in} /
\mathscr{V}_{P,{\rm ull}}^{\rm out/in}.$$
See \cite{KugoOjima, Nakanishi}.

{\bf (iii)}
Next we determine the physical scattering operator.
Consider the scattering operator
$$S_P:
{\mathscr{V}_{P,{\rm phys}}^{\rm out}}\rightarrow
{\mathscr{V}_{P,{\rm phys}}^{{\rm in}}}$$ as a unitary operator;
namely $S_P{\mathscr{V}_{P,{\rm phys}}^{\rm out}}=
{\mathscr{V}_{P,{\rm phys}}^{{\rm in}}}$.
We can check that $S_P$ leaves the null space invariant
and define the
physical scattering operator
$S_{P,{\rm phys}}$ by
$$S_{P,{\rm phys}}[\Psi]_{\rm out}:=[S_P\Psi]_{\rm in}$$
for $[\Psi]_{\rm out/in}\in \hhh _{P,{\rm phys}}^{\rm out/in}$.
It can be shown that
this is also a unitary operator
from
$\hhh _{P,{\rm phys}}^{\rm out}$ to
$\hhh _{P,{\rm phys}}^{\rm in}$ in Theorem \ref{suzu2}.

{\bf (iv)}
It can be seen from Lemma \ref{equivH}
that
$$H_{P,{\rm phys}}^{\rm out/in}=
\left[
H_P\lceil_{D(H_P)\cap \V_{P,{\rm phys}}^\ex} P^{\oi}\right]_\ex$$ is
a well defined operator on $\hhh _{P,{\rm phys}}^\ex$, where $P^\oi$ denotes the projection onto $\V_P^\oi$.
We call this the physical Hamiltonian.
It is proven in Theorem \ref{suzu1}
that
$H_{P,{\rm phys}}^{\rm out/in}$
is a  {\it self-adjoint operator}
on $\hhh _{P,{\rm phys}}^{\rm out/in}$.
Note that the physical Hamiltonian $H_{P,{\rm phys}}^{\rm out/in}$ is self-adjoint, whereas
our Hamiltonian $H_P$ is {\it not} self-adjoint,

This paper is organized as follows.
In Section 2 we define NRQED.
In Section 3
we present the explicit form
of the Heisenberg operators.
In Section 4 we construct the asymptotic fields $\A_\mu^{\rm out/in}(f,t,P)$ based on the LSZ formalism
 and define the scattering operator $S_P$.
 In section 5 we define
 physical subspaces  in an abstract way,
  and characterize the physical subspace at time
  $t$ and for $t=\pm\infty$.
In Section 6 we define the physical Hamiltonian
on the quotient space $
\hhh _{P,{\rm phys}}^{\rm out/in}$.

\section{NRQED in the Lorentz gauge}
\subsection{Boson Fock space}
We begin by defining elements of a Boson Fock space.
The Boson Fock space over
the Hilbert space $\oplus^4\LR$ is given by
the infinite direct sum of the $n$-fold symmetric tensor product of $\oplus^4\LR$:
\eq{fock}
\ffff:=\fff(\oplus^4\LR)=
\bigoplus_{n=0}^\infty
\lk
\bigotimes_s^n
\lk
\oplus^4  \LR
\rk \rk.
\en
Here $\otimes_s^n$ denotes the symmetric tensor product and we set $\otimes_s^0(\oplus^4\LR)=\CC$. We denote
the scalar product on $\ffff$ by
\eq{norm} {(\Phi,\Psi)_\ffff}:=\sum_{n=0}^\infty
(\Phi^{(n)},\Psi^{(n)})_{\otimes^n
(\oplus^4\LR)} \en  for $\Psi=\{\Psi^{(n)}\}_{n=0}^\infty$ and
$\Phi=\{\Phi^{(n)}\}_{n=0}^\infty\in\ffff$,
 which is
  anti-linear in $\Phi$ and linear in $\Psi$.
Then $\ffff$ becomes a Hilbert space with scalar product given by
\kak{norm}.

The creation operator
$\add(F):\ffff\rightarrow \ffff$ with a smeared function
 $F\in \oplus^4\LR$
is defined by
\eq{crea}
\lk
\add(F)\Psi\rk^{(n+1)}=\sqrt{n+1}S_{n+1}(F\otimes \Psi^{(n)}),\quad n\geq 0,
\en
with domain
$$D(\add(F))=\lkk
\{\Psi\}_{n=0}^\infty\in\ffff
\left|
 \sum_{n=0}^\infty (n+1)\|S_{n+1}
 (F\otimes \Psi^{(n)})\|^2<\infty\right.\rkk,$$
where $S_n$ denotes the symmetrizer defined by
$S_n(F_1\otimes\cdots\otimes F_n)=(n!)\f \sum_{\pi\in S_n}F_{\pi(1)}\cdots\otimes F_{\pi(n)}$ with $S_n$
the set of
permutations of degree $n$.
The annihilation operator $a(F)$ is defined by
the adjoint of
$\add(\bar F)$ with respect to the scalar product  \kak{norm},
i.e., $a(F)=\lk \add(\bar F)\rk ^\ast$.
We can identify $\ffff$ as
 \eq{id}
\ffff\cong   \ffff_1\otimes \ffff_2
\otimes\ffff_3\otimes\ffff_0,
\en
where $\ffff_\mu=\fff(\LR)$, $\mu=0,1,2,3$.
Hereinafter we make this identification
without  further notice,
 and under this identification we set
 $$a(f,\mu):=\lkk
\begin{array}{ll}
a(f,1)=a(f)\otimes 1\otimes 1\otimes 1,\\
a(f,2)=1\otimes a(f)\otimes 1\otimes 1,\\
a(f,3)=1\otimes 1\otimes a(f)\otimes 1,\\
a(f,0)=1\otimes 1\otimes 1\otimes a(f).
\end{array}
\right.
$$
We also define
$\add(f,\mu)$ in a similar manner
and
formally write
$$
\ass(f,\mu)=\int \ass(k,\mu) f(k) dk,\qquad a^\sharp=a,\add,
$$
with the informal kernel $\ass(k)$.

$\Omega\in \ffff$ denotes  the Fock vacuum defined by $\Omega=\{1,0,0,...\}$.
 The Fock vacuum  is the unique vector
such that $a(f,\mu)\Psi=0$
   for all $f\in\LR$ and $\mu=0,1,2,3$.
Let $\Omega_\mu \in \ffff_\mu$ be the Fock vacuum. Then $\Omega = \Omega_1 \otimes \Omega_2 \otimes \Omega_3 \otimes \Omega_0$ follows.
The set of vectors
$$\fffff:={\rm L.H.}\left\{
\left.
\PI \add(f_i,\mu_i)\Omega,~
\Omega\right|f_j\in\LR,\mu_j=0,1,2,3\right\}$$
is called the finite particle subspace and it is dense in $\ffff$.
The annihilation and creation operators
leave $\fffff$ invariant
and
satisfy the canonical commutation relations:
\eq{ccr}
[a(f,\mu),\add(g,\nu)]=\delta_{\mu\nu}
(\bar f, g),\quad
[\ass(f,\mu),\ass(g,\nu)]=0,
\quad \mu,\nu=0,1,2,3, \en
on $\fffff$.

Next we define the second quantization.
Let $\mathscr{C}({\cal K})$ denote
the set of contraction operators on Hilbert space
${\cal K}$.
The functor $\Gamma:\mathscr{C}(\oplus^4 \LR)\rightarrow
  \mathscr{C}(\ffff)$ is
  given by
  $$\Gamma(T)
\PI \add(F_i) \Omega=
\PI \add(TF_i)\Omega$$
and $\Gamma(T)\Omega=\Omega$
for
$T\in \mathscr{C}(\oplus^4 \LR)$, which is called the second quantization of $T$.
Let
\eq{formfactor}
\omega(k)=|k|,\quad k\in\BR.
\en
The second quantization of the one-parameter multiplicative unitary group $
e^{-it\omega}$ on $\LR$ induces the one-parameter unitary group
$\{\Gamma(\oplus^4e^{-it\omega})\}_{t\in\RR}$
on $\ffff$.
Its self-adjoint generator is denoted
by $\hf$, i.e.,
\eq{hf}
\Gamma(\oplus^4 e^{-it\omega})=e^{-it\hf},\qquad t\in\RR,
\en
 and it is formally written as
\eq{hfff}
\hf=\sum_{\mu=0}^3 \int \omega(k)\add(k,\mu)a(k,\mu)dk.
\en
Replacing $\omega(k)$ with the multiplication by
the identity $1$ in \kak{hfff},
we define the number operator $N_{\rm f}$ of $\ffff$.
Furthermore let
	\[ N^0_{\rm f} = \int \add(k,0)a(k,0)dk. \]
This is the number operator on $\ffff_0$.

\subsection{Indefinite metric}
Let $g$ be the $4 \times 4$ matrix $g=(g_{\mu\nu})_{\mu,\nu=0,1,2,3}$
given by
\eq{mink}
 g_{\mu\nu}=\lkk\begin{array}{rl}1,&\mu=\nu=0,\\
-1,&\mu=\nu\not=0,\\
0,&\mu\not=\nu.
\end{array}
\right.
\en
 Now we  introduce  the indefinite-metric on $\ffff$.
  Let $[g]$ be the linear operator induced from the metric tensor
$g$:
$$[g]=\lk\begin{array}{cccc}
-1&0&0&0 \\
0&-1&0&0\\
0&0&-1&0\\
0&0&0&1 \end{array} \rk:\oplus^4\LR\rightarrow \oplus^4\LR,
$$
and define $\eta:\ffff\rightarrow \ffff$ by the second quantization of $-[g]$,  i.e.,
\eq{ga}
\eta:=\Gamma(-[g])= (-1)^{N^0_{\rm f}}.
\en
From the definition,
\eq{eta2}\eta^2=1.
\en
By using $\eta$ we introduce
an indefinite metric on $\ffff$ by
 \eq{inde}
 \im{\Psi}{\Phi}:=(\Psi,\eta\Phi)_\ffff.
 \en
In order to define the adjoint with respect to
the indefinite metric \kak{inde}
we  introduce  the $\eta$-adjoint
of $a(f,\mu)$ by
\eq{eta}
\addd(f,\mu):=\eta \add(f,\mu)\eta.
\en
Then  $\im{\Psi}{a^\dagger(f,\mu)\Phi}=
\im{a(\bar f,\mu)\Psi}{\Phi}$ and
\eq{ad2}
\addd(f,\mu)=\lkk\ar{rl}{\add(f,j),&\mu=j=1,2,3,\\
-\add(f,0),&\mu=0}\right.
\en
hold.
Hence we have the commutation relations:
\eq{ccr22}
[a(f,\mu),\addd(g,\nu)]=-g_{\mu\nu}(\bar f,g),\qquad
[\addd(f,\mu),\addd(g,\nu)]=0.
\en

Let us define  the quantized radiation field
$\A _\mu(f,x)$, $x\in\BR$, $\mu=0,1,2,3$,
 for a test
function $f\in\LR$.
Let $e^j(k)\in\BR$, $k\in\BR$, $j=1,2,3$,
 be unit vectors such that
 $e^3(k)=k/|k|$, and let
 three vectors $e^1(k)$, $e^2(k)$ and $e^3(k)$
 form a right-handed system for each $k\in\RR^3$.
We fix them.
The  quantized
 radiation field,
 \eq{rad}
 \left(\skima
 \A_0(f,x),-\A_1(f,x),-\A_2(f,x),-\A_3(f,x)
 \right)=
 \left(\A^0(f,x),\vA(f,x)\right),
 \en
 smeared by the test function
 $f\in \LR$ at time zero
  is defined by
 \eqn && \vA (f,x)=\frac{1}{\sqrt
2}\jjj\int dk \frac{e^j(k)}{{\sqrt{\omega(k)}}}
 \lk \addd(k,j){\hat f(k)}e^{-ikx}+
a(k,j){\hat f(-k)}e^{ikx}\rk,\\
&& \A _0(f,x)=\frac{1}{\sqrt 2}\int dk \frac{1}{{\sqrt{\omega(k)}}}
 \lk \addd(k,0){\hat f(k)}e^{-ikx}+
 a(k,0){\hat f(-k)}e^{ikx}\rk \enn
and
its derivative by
 \eqn
&& \vec{\dot{\A }}
(g,x)= \frac{i}{\sqrt 2}\jjj\int dk
e^j(k)\sqrt{\omega(k)}\lk  \addd(k,j)
 {\hat g (k)}e^{-ikx}
- a(k,j) {\hat g (-k)}e^{ikx}\rk,\\
&&
\dot{\A }_0(g,x)=
\frac{i}{\sqrt 2}\int dk
\sqrt{\omega(k)}\lk  \addd(k,0)
 {\hat g (k)}e^{-ikx}
- a(k,0) {\hat g (-k)}e^{ikx}\rk.
\enn
$\A^0$ is a scalar potential and $\vA$ a vector potential.
Conventionally the vector potential $\vA$ is decomposed as
$\vA=\A^\perp+\A^{\parallel}$, where $\A^\perp$
is the transversal part and $\A^\parallel$ the longitudinal part given by
\begin{eqnarray*}
&& \A^\perp(f,x)
=\frac{1}{\sqrt
2}\sum_{j=1,2}\int dk \frac{e^j(k)}{{\sqrt{\omega(k)}}}
 \lk \addd(k,j){\hat f(k)}e^{-ikx}+
a(k,j){\hat f(-k)}e^{ikx}
\rk,\\
&&
\A^{\parallel}(f,x)
=\frac{1}{\sqrt2}
 \int dk \frac{e^3(k)}{{\sqrt{\omega(k)}}}
 \lk \addd(k,3){\hat f(k)}e^{-ikx}+
a(k,3){\hat f(-k)}e^{ikx}\rk.
\end{eqnarray*}
Note that $\mmml \partial_{x_l}\A_l^\perp(f,x)=0$.
Set
\eq{ad}
\A_\mu(f):=\A_\mu(f,0),\qquad \mu=0,1,2,3.
\en
By the canonical commutation
relations \kak{ccr} and \kak{ccr22} we have
\eq{sa8}
[\A _\mu(f),\dot{\A }_\nu(g)]=-ig_{\mu\nu}(\bar f, g)
\en
 and
\eq{sa9}
[\A _\mu(f), \A _\nu(g)]=0,\quad
 [\dot{\A }_\mu(f),\dot{\A }_\nu(g)]=0
 \en
 for {\it all} $f,g\in\LR$.
It can also be seen that
\eq{ha}
[\hf,\addd(f,\mu)]=\addd(\omega f,\mu),\quad
[\hf,a(f,\mu)]=-a(\omega f,\mu),
\en
by which
we have
\eq{ap}
[\hf, \A _\mu(f)]=-i\dot{\A }_\mu(f),\quad
[\hf, \dot{\A }_\nu(f)]=i\A _\mu(-\Delta f).
\en

We introduce notions of $\eta$-self-adjointness and $\eta$-unitarity \cite{Bogner} below.
\begin{definition}
\begin{enumerate}
\item[(1)] A densely defined linear operator $X$ is $\eta$-self-adjoint if and only if  $\eta X^\ast \eta=X$.
\item[(2)] A densely defined linear operator $X$ is $\eta$-unitary if and only if $X$ is injective and
$X^{-1}=\eta X^* \eta$.
\end{enumerate}
\end{definition}
The next lemma immediately follows from the definition of $\eta$-self-adjointness.
\bl{etasym}
(1) $X$ is $\eta$-self-adjoint if and only if
$\eta X$ is self-adjoint.
(2) Let $X$ be $\eta$-self-adjoint. Then
$X$ is closed on $D(X)$.
(3) Let $X$ be $\eta$-self-adjoint and $\eta X$ is essentially self-adjoint on $D$.
Then $D$ is a core of $X$.
\el
For real-valued $f$, note that
the closures of  $\A _j(f,x)$ and $\dot{\A }_j(f,x)$, $j=1,2,3$,   are self-adjoint
and $\eta$-self-adjoint for each $x\in\BR$.
However the closure of $\A _0(f,x)$ and $\dot{\A }_0(f,x)$ for real-valued $f$ are $\eta$-self-adjoint but
not even symmetric.
Moreover
the free Hamiltonian
 $\hf$ is self-adjoint and $\eta$-self-adjoint.

\subsection{Definition of NRQED in the Lorentz gauge}
The Hilbert space of our system consisting of one electron coupled with photons
is given by
the tensor product of $\LR$ and $\ffff$:
\eq{hi}
\hhh:=\LR\otimes\ffff,
\en
where $\LR$ describes the state space of one electron and $\ffff$ the photon field.
The full Hamiltonian of our system is defined by
\eq{full}
H_{\rm tot}:=
\frac{1}{2m}\left(p \otimes 1-e \vA(\vp,\cdot)\right)^2+1\otimes \hf+e 1\otimes \A_0
\en
for a given fixed test function $\vp$ on $\BR$ which satisfies some conditions mentioned later.
Let $m>0$ and $e\in\RR$
denote the mass and charge of the electron, respectively, and $p =-i\vec\nabla_x$ denote the momentum operator of the electron.
 Instead of this full Hamiltonian in this paper
we take the dipole approximation;
namely we replace $\A(\vp,\cdot)$ in $H_{\rm tot}$ by $1\otimes \A(\vp)$.
We set
\eq{rad2}
\A _\mu := \A _\mu(\vp).
\en
We make the following assumptions about $\vp$ throughout this paper.
\begin{assumption}\label{ass1}
{\rm (Assumptions for $\eta$-self-adjointness)}.
$\vp/\omega,\vp/\w, \w\vp\in\LR$ and $\vp(-k)=\ov{\vp(k)}$.
\end{assumption}
Then our Hamiltonian
is given by
\eq{sa1}
H:=\frac{1}{2m}\left(p \otimes 1-e1\otimes \vA\right)^2+1\otimes \hf+e 1\otimes \A_0
\en
with domain
\eq{rad21}
D(H):=D(-\Delta\otimes 1)\cap D(1\otimes \hf).
\en
\bp{etaself}
$H$ is $\eta$-self-adjoint
and
$\eta H$ is essentially self-adjoint on any core of
$-\Delta\otimes 1+1\otimes \hf$.
In particular
$H$ is closed
and an arbitrary core of $-\Delta\otimes 1+1\otimes \hf$
is also a core of $H$.
\ep
\proof
Set $H'=H-e1\otimes \A_0$.
Let $L=-\Delta\otimes 1 +1\otimes \hf+1$.
Then
we have
$$|(L\Psi, H'\Phi)-(H'\Psi, L\Phi)|\leq C \|L^\han \Phi\| \|L^\han \Psi\|$$ for some constant $C$ by
the fundamental inequality
$\|\ass(f)\Psi \|\leq C'\|(\hf+1)^\han \Psi\|$.
Thus by the Nelson commutator theorem, $H'$
is self-adjoint on $D(-\Delta\otimes 1)\cap D(1\otimes\hf)$.
We can also see that
\eq{kouta1}
\|\eta \A_0\Psi\|\leq C'(\|(\hf^0)^\han\Psi\|+\|\Psi\| )
\en
and $[H',\eta]=0$, which
implies that
$$
\|\eta \A_0\Psi\|\leq
C'(\|(H')^\han\Psi\|+\|\Psi\| )\leq
\epsilon\|\eta H'\Psi\|+b_\epsilon \|\Psi\|$$
for arbitrary $\epsilon>0$.
Since $\A_0$ is skew symmetric and $\{\A_0, \eta\}=0$,
we have $(\eta H)^*=H^* \eta \supset \eta H$,
which yields the result that $\eta H$ is symmetric.
Then we can see by the Kato-Rellich theorem
that $\eta H$ is self-adjoint on
$D(-\Delta\otimes 1)\cap D(1\otimes \hf)$.
This completes the proof.\qed
We divide $\hhh$
into a scalar part and a vector part.
Let
$
\hhh_0:=\ffff_0
$,
$\ffff_{\pp}=\ffff_1\otimes\ffff_2\otimes\ffff_3$
and
$\hhh_{\pp}:=\LR\otimes\ffff_{\pp}$.
Then
$\hhh$ can be realized as the tensor product of
the scalar part and the vector part:
\eq{iden}\hhh\cong\hhh_{\pp} \otimes \hhh_{0}.
\en
We use this identification without further notice through this paper.
This identification is inherited by the Hamiltonian $H$ and
we have
 \eq{h}
 H= H_{\pp}\otimes 1 + 1\otimes H_{0},
  \en
 where $H_{\pp}$ is the vector component of $H$:
 \eq{hz}
 H_{\pp}:=
 \frac{1}{2m}(p  \otimes 1 -e1\otimes \vA )^2+1\otimes
\hff
 \en
defined on $\hhh_{\pp}=\LR\otimes\ffff_{\pp}$,
and $H_0$ the scalar component:
\eq{ht}
H_0:=
 e \A _0+\hf^0
 \en
 defined on $\hhh_0$.
 Here $\hff$ denotes the free Hamiltonian  in $\ffff_{\pp}$:
  $$\hff=
  \jjj\int\omega(k)\add(k,j)a(k,j)dk=
  \jjj\int\omega(k) a^\dagger (k,j)a(k,j)dk
    $$
and $\hf^0$ in $\ffff_0$:
$$\hf^0=
 \int\omega(k)\add(k,0)a(k,0)dk=
-\int\omega(k)\addd(k,0)a(k,0)dk.
$$
\bp{self2}
(1)
$H_\pp$ is self-adjoint on $D(-\Delta\otimes 1)\cap D(1\otimes \hf^\pp)$ and essentially self-adjoint
on any core of
$-\Delta\otimes 1+ 1\otimes \hf^\pp$.
(2)
$H_0$ is $\eta$-self-adjoint on $D(\hf^0)$. In particular
$H_0$ is closed on $D(\hf^0)$ and an arbitrary
 core of $\hf^0$ is also a core of $H_0$.
\ep
\proof
(1) has been proven in the proof of Proposition \ref{etaself}.
By \kak{kouta1},
$\eta H_0$ is self-adjoint on $D(\hf^0)$. Hence the proof is complete.
\qed

\section{Heisenberg equations}
In this section we first
 diagonalize the total Hamiltonian by making use of a certain $\eta$-unitary operator, and solve the Heisenberg equation exactly.
The first rigorous results on the diagonalization of NRQED are in
Arai  \cite{a83,a83b}, where the electromagnetic field is quantized with respect to the Coulomb gauge and then there is no  scalar potential $\A_0$
nor longitudinal potential $\A^\parallel$.

In addition to Assumption \ref{ass1},
from now on we make the following assumption.
\begin{assumption}
\label{ass2}
We assume (1)-(5).
\begin{description}
\item[(1)]
   $\int_\BR|\vp(k)|^2/\omega(k)^3 dk<\infty$,
   \item[(2)]
    there exists $\epsilon>0$ such that $\|e^{+\epsilon\omega} \vp\|_\infty<\infty$,
   \item[(3)]
    there exists a function $\rho$ on $[0,\infty)$ such that $\vp(k)=\rho(|k|)$,
   \item[(4)]
    $\rho(s)>0$ for $s\not=0$,
   \item[(5)]
    $F(s):=\rho(\sqrt s)^2\sqrt s \in L^p([0,\infty);ds)$ for some $1<p$,  and there exists $0<\alpha <1$ such that $|F(s+h)-F(s)|\leq K|h|^\alpha$ for all $s$ and $0<h\leq 1$.
   \end{description}
   \end{assumption}
We explain Assumption \ref{ass2}.

(1) This condition is called the infrared regular condition,
which is used to construct $\eta$-unitary operators $V_0$ in Section 3.1 and $U_P$ in Section 3.2.

(2) This ensures that $\|\omega\vp\|_\infty<\infty$, $\|\sqrt\omega\vp\|_\infty<\infty$ and $\|\vp\|_\infty<\infty$. Then
the operators $T$ and $W_\pm^{ij}$ introduced in Section 3.2 can be defined as  bounded operators.
Furthermore  we can construct the Heisenberg operators defined by Definition \ref{DHO+} by (2), where we need analytic continuation of $e^{-it\omega}\vp$ with respect to $t$.

(3) This means that $\vp$ is rotation invariant.

In Section 3.2 we introduce the function \eq{dead}
D_\pm(s)=m-\frac{e^2}{2}4\pi\lkk
\lim_{\epsilon\rightarrow 0}\int _{|s-r|>\epsilon, r>0}\frac{\rho(\sqrt r)^2 \sqrt r}{s-r}dr
\mp 2\pi i \rho(\sqrt s)^2 \sqrt s\rkk
\en
and define
\eq{qdef}
Q(k)=\frac{\vp(k)}{D_+(\omega(k)^2)}.
\en
Conditions (4) and (5) are used to ensure that $Q$ is well defined.

(4)
This condition implies that the imaginary part of $D_\pm$ is strictly positive for all $s>0$, in particular it is non-zero except for $s=0$.

(5) Note that
$$HF(s)=(2\pi i)^{-1}\lim_{\epsilon\rightarrow0}\int_{|s-r|>\epsilon, r>0} \frac{F(s)}{s-r} dr$$ is the
Hilbert transform of $F(s)$.
By the Lipshitz condition (5)
the real part of $D_\pm$ is also Lipshitz continuous with the same order $\alpha$ as $F(s)$  and belongs to $L^p(\RR)$. See e.g.,  \cite[p.145, THEOREM 106]{tit}.
This yields the result that
   the real part of $D_\pm(s)\rightarrow 0$ as $s\rightarrow \infty$.

Thus (4) and (5), together with $D_\pm(0)>0$,  ensure that
there exists $c>0$ such that
\eq{dpos}
|D_\pm(s)|>c,\quad \forall s\geq 0
\en
and hence $Q$ is well defined.

Assumption (4), however, seems to be unusual.
We note that, as mentioned above,
assumptions  (4) and (5) are
sufficient to allow the definition of $Q$.
It is possible to choose an alternative $\rho$ so that $Q$ is well defined.
In particular, one can choose $\rho$ satisfying (5) but not (4).
For example suppose that $\rho(s)$ has compact support $|s|<N$.
Then in order to define $Q$ it is enough to
assume further that
${\rm Re} D_+(s)=
m-(2\pi e^2)( 2\pi i) HF(s)>\delta>0$ for all $|s|>N-1$ for some $\delta>0$.  Notice that this assumption is satisfied, because $HF(s)\rightarrow 0$ as $s\rightarrow \infty$ and $m>0$.  We omit the detail.

\subsection{Scalar potential}
Let us begin by discussing the scalar part of the Hamiltonian.
The scalar part $H_0$ of $H$ can be easily diagonalized  by an $\eta$-unitary (but not unitary) operator. Let
\eq{vz}
V_0 = \exp
\left(
 \frac{e}{\sqrt{2}}
			 \left(a^\ast \left(\frac{\hat{\varphi}}{\omega^{3/2}},0\right)
				 +a\left(\frac{\hat{\varphi}}{\omega^{3/2}},0\right) \right) \right) .
\en
This is unbounded $\eta$-unitary on $\ffff_0$.
It can be seen that the finite particle subspace $\ffff_{0, \rm fin}$ of $\ffff_0$
contains analytic vectors of $V_0$. Then
\eq{ana1}
 V_0 \Psi=\sum_{n=0}^\infty
\frac{1}{n!}
{
\left(a^\ast \left(\frac{\hat{\varphi}}{\omega^{3/2}},0\right)
				 +a\left(\frac{\hat{\varphi}}{\omega^{3/2}},0\right) \right) ^n
}\Psi
\en
for $\Psi\in \ffff_{0, \rm fin}$.
By the commutation relations and \kak{ana1}, we have
$$(H_0^\ast \Phi, V_0\Psi)=(\Phi, V_0(\hf^0+E_0)\Psi)$$
 for  $\Psi,\Phi\in
\ffff_{0, \rm fin}$,
where
\eq{ep}
E_0 := \frac{e^2}{2}\|{\hat{\varphi}}/{\omega}\|^2.
\en
Thus $V_0\Psi\in D(H_0)$ and
$H_0V_0\Psi= V_0(\hf^0+E_0)\Psi$; furthermore
$H_0V_0\Psi\in D(V_0\f)$. Then
we have
\eq{dia}
V_0\f H_0 V_0=\hf^0+E_0
\en
on $\ffff_{0, \rm fin}$. Since $\ffff_{0, \rm fin}$ is a core of $\hf^0$,
we
have
\eq{dia2}
\ov{V_0\f H_0 V_0\lceil_{\ffff_{0, \rm fin}}}=\hf^0+E_0
\en
on $D(\hf^0)$.
From \kak{dia}, we can see that
$V_0\Omega_0$ is an eigenvector of $H_0$
associated with eigenvalue $E_0$:
\eq{eigen}
H_0V_0\Omega_0=E_0V_0\Omega_0.
\en
Define
\eqnn
 &&\label{b100}
 b(f,0):=a(f,0)-\frac{e}{\sqrt2} (\vp/\omega^{3/2},f),\\
  &&\label{b200}
 b^\dagger(f,0):=a^\dagger (f,0)-\frac{e}{\sqrt2}
 (\vp/\omega^{3/2},f).
 \ennn
These operators satisfy
the canonical commutation relations:
\eq{bbcan}
[b(f,0),b^\dagger (g,0)]= -(\bar f, g),\quad [b(f,0),b(g,0)]=0 = [b^\dagger (f,0),b^\dagger(g,0)]
\en
and
\eq{bbh}
[H_0,b^\dagger(f,0)]=b^\dagger(\omega f,0),\quad
[H_0, b(f,0)]=-b(\omega f,0).
\en
Thus
the quadruple
\eq{sa45}
(\ffff_0, V_0\Omega_0, \{b^\dagger (f,0),b(f,0)|f\in\LR\}, H_0)
\en
corresponds to the free case
$(\ffff_{0}, \Omega_{0}, \{a(f,0),\add(f,0)|
f\in\LR\},\hf^0)$,
but $H_0$ is not self-adjoint.

\subsection{Vector potential}
In this subsection we investigate the vector part $H_{\pp}$ of $H$.
$H_{\pp}$ is quadratic and can also be diagonalized by a Bogoliubov transformation.

$H_{\pp}$ is self-adjoint on $D(-\Delta\otimes1)\cap D(1\otimes \hf)$
and essentially self-adjoint on any core of $(-1/2)\Delta\otimes 1+1\otimes\hf$. This can be proven by virtue of the Nelson commutator theorem as stated in the proof of Lemma \ref{etaself}.
 Since $H_{\pp}$ commutes with
$p_j$, $j=1,2,3$,
$H_{\pp}$ and $\hhh_{\pp}$ are decomposable with respect to the
spectrum of $p_j$ and are given by
$$
\hhh_{\pp} =\int^\oplus_\BR \hhh_{\pp,P}
dP$$
and
$$
H_{\pp}=\int^\oplus_\BR H_{\pp,P} dP,$$
where
$\hhh_{\pp,P}=
\ffff_{\pp}$
and
$H_{\pp,P}$ is  the self-adjoint operator
on $\fff_{\pp}$, given by
\eq{fiber}
H_{\pp,P}=
\frac{1}{2m}(P-e \vA )^2+ \hff,\quad P\in\BR.
\en
The fiber Hamiltonian $H_{\pp,P}$
  is, indeed,  self-adjoint on $D(\hff)$ for all $(P,e)\in\BR\times\RR$ and
  bounded from below.
In the similar way as Proposition \ref{etaself} this can also be proven by virtue of the Nelson commutator theorem with the conjugate operator $L$ replaced by $N_{\pp}+1$, where $N_{\pp}$ denotes the number operator on $\ffff_{\pp}$.
Now for each $(P,e)\in\BR\times\RR$,
let us construct a quadruple \kak{omega} relevant to
the free case $(\ffff_{\pp}, \Omega_{\pp}, \{a(f,j),\add(f,j)|f\in\LR,j=1,2,3\},\hf^{\pp})$:
 \eq{omega}
\left (
\ffff_{\pp}, \Omega_{\pp, P}, \{b_P(f, j),
\bdd(f, j)|f\in\LR,j=1,2,3\}, H_{\pp,P}
\right)
 \en
   such that
\bi
\item[(1)]
$b_P(f, j)$ and $\bdd(g, j)$ satisfy the canonical commutation relations,
$$[b_P(f,j),\bdd(g,i)]=
\delta_{ij}(\bar f,g),\quad\quad [b_P(f,j),b_P(g,i)]=0=
[\bdd (f,j),\bdd(g,i)]
,$$
\item[(2)]
 $[H_{\pp,P},b_P(f, j)]=-b_P(\omega f, j)$ and
   $[H_{\pp,P},\bdd(f, j)]=\bdd(\omega f, j)$,
  \item[(3)]
   $\Omega_{\pp, P}$ is the unique vector such that
$b_P(f, j)\Omega_{\pp, P}=0$ and
is the ground state of $H_{\pp,P}$.
\ei
 From (1) to (3) above we will be able to
 infer
    the unitary equivalences:
   $\Omega_{\pp, P}\cong \Omega$,
 $\bss_P(f, j)\cong \ass(f)$ and
 $H_{\pp,P}\cong \hff+E_{\pp}(P)$ for each $P$,
  where $E_{\pp}(P)=\is(H_{\pp,P})$ is  given explicitly.

In order to construct $b_P^\sharp$ we make explicit the relationship between $\ass$ and $\vA$.
It can be seen that the creation operator
and the annihilation operator
can be expressed as
\eqnn &&
\label{asss1}
a(f,l)=\frac{1}{\sqrt{2}}\mmm  \lk \hat \A _j(e_j^l\sqrt\omega f)+i\hat{\dot{\A}}_j
(e_j^l\ww f)\rk,\\
&&\label{asss2}
 \add (f,l)=
 \frac{1}{\sqrt{2}} \mmm \lk \hat \A _j(\tilde e_j^l\sqrt\omega \tilde f)-i\hat{\dot{\A}}_j(\tilde e_j^l \ww \tilde f)\rk,
 \ennn
where $\hat \A (f):=\A (\hat f)$ and $\hat{\dot{\A}}(g):=\dot{\A}(\hat g)$,  and $\tilde f(k)=f(-k)$. Note that $\hat{\hat f}=\tilde{\check{\hat f}}=\tilde f$.
Modifying the right-hand side of \kak{asss1} and \kak{asss2},
we can construct $b^\sharp$ in \kak{omega}.
 Let
$$G_\epsilon f(k):=\int_\BR \frac{f(k')}{(\omega(k)^2-\omega(k')^2+i\epsilon)\omega(k)^\han\omega(k')^\han}dk',\quad \epsilon>0.$$
Then $G_\epsilon$ is bounded and skew-symmetric on $\LR$. Moreover the strong limit $G:=\lim_{\epsilon\downarrow 0}G_\epsilon$ exists
as a bounded skew-symmetric operator.

Let
\eq{D}
D(z):=m-e^2\int_{\BR}\frac{|\vp(k)|^2}
{z-\omega(k)^2}dk,
\en
which is analytic on $\CC\setminus[0, \infty)$.
Let
$D_{\pm}(s):=\lim_{\epsilon\downarrow0}D(s\pm i\epsilon)$ for $s\in [0,\infty)$; then
we see that $|D_{\pm}(s)|>c$ for some $c>0$ by
(4) and (5) of Assumption \ref{ass2}.
See \kak{dpos} and
\kak{dead} for the explicit form of $D_\pm$.
Then we can define
$Q(k):={\vp(k)}/{D_+(\omega(k)^2)}$.
Operator $T:\LR\rightarrow\LR$ is given by
\eq{TT}
Tf:=f+e^2Q\sqrt\omega G \sqrt \omega \vp f.
\en
Since $G$ is skew-symmetric, we have
$T^\ast f=f-e^2\vp\w G\w \bar Q f$.
\bl{1}
$T$ satisfies the following algebraic relations:
\bi
\item[(1)]
$T$ is unitary on $\LR$ and bounded on $L^2(\BR,\omega^n dk)$, $n=\pm1$;
\item[(2)]
$\d  T^\ast \frac{1}{\omega^2}Q
     =\frac{\vp}{\mass\omega^2}$,
    where $\mass:=D(0)=m+e^2\|\vp/\omega\|^2$;
    \item[(3)] $[\omega^2,T^\ast ]f=-e^2(Q,f)\vp$,\quad  $[\omega^2,T]f=+e^2(\vp,f)Q$;
        \item[(4)] $T\vp=mQ$.
    \ei
    \el
\proof This is a slight modification of \cite{a83,a83b}.
We omit the proof.
 \qed
    Now, for $f\in\LR$, we define
\eqnn
&&\label{b1}
 b_P(f,l):=\frac{1}{\sqrt{2}}
\mmm \lk \hat \A _j(T^\ast e_j^l\sqrt\omega f)+i\hat{\dot{\A}}_j(T^\ast e_j^l\frac{1}{\sqrt\omega} f)-P_j\lk
\frac{ee_j^l Q}{\omega^{3/2}},f\rk
\rk,\ \ \ \\
&&\label{b2}
 \bdd (f,l):=
\frac{1}{\sqrt{2}} \mmm \lk \hat \A _j(\bar T^\ast \tilde e_j^l\sqrt\omega \tilde f)-i\hat{\dot{\A}}_j(\bar T^\ast \tilde e_j^l\frac{1}{\sqrt\omega} \tilde f)- P_j\lk
\frac{ee_j^l \bar Q}{\omega^{3/2}},f\rk
 \rk \ennn
and set
$b^\sharp (F):=\mmml b^\sharp (F_l,l)$
for $F\in\oplus^3\LR$.
\bl{2}
It follows that
$\lk b_P(f,j)\rk ^\ast=b_P^\ast(\bar f,j)$,
 and
the commutation relations below hold:
\eqnn
&&\label{ccr2}
[b_P(f,i),\bdd(g,j)]=\delta_{ij}(\bar f,g),\quad
[b_P(f,j),b_P(g,i)]=0=[b_P^\ast (f,j),b_P^\ast (g,i)],\quad \quad \\
&&\label{hccr}
[H_{\pp,P}, b_P(f,j)]=-b_P(\omega f,j),\quad
[H_{\pp,P},\bdd(f,j)]=\bdd(\omega f,j).
\ennn
\el
\proof
 By the definition of $\bss_P$ we have
\eqn
&&
b_P(f,j)=\iii
\lk \add(W_-^{ij} f,i)+a(W_+^{ij}f,i)+
\sum_{l=1}^3 (P_l  L_l^j,f)\rk,\\
&&
  \bdd(f,j)=\iii\lk \add(\bar W_+^{ij}f,i)+a(\bar W_-^{ij}f,i)+\sum_{l=1}^3 (P_l \bar L_l^j,f)\rk,
  \enn
   where $\bar X f =\ov{X\bar f}$, $\d L_l^j =e\frac{1}{\sqrt2} \frac{e_l^j Q}{\omega^{3/2}}$ and
   $W_{P,\pm}^{ij}:\LR\rightarrow \LR$ is defined by
   \eqn
&& W_+^{ij}f:=\half
\sum_{l=1}^3
e_l^i
\lk
\ww  T^\ast  \w
+
\w  T^\ast \ww \rk e^j_l f,\\
&& W_-^{ij}f:=\half
\sum_{l=1}^3
e_l^i \lk
\ww  T^\ast  \w
-
\w  T^\ast \ww \rk
\tilde e_l^j
 \tilde f.
\enn
 Then
$W_{\pm}=\lk W_{\pm}^{ij}\rk_{1\leq i,j\leq 3}:\oplus^3\LR\rightarrow \oplus^3\LR$
 has the symplectic structure \kak{symplectic2} below.
Let \eq{symplectic}
{\Bbb W}=\MM {W_+}{\bar W_-}{W_-}{\bar W_+}:\bigoplus^2[\oplus^3\LR]\rightarrow \bigoplus^2[\oplus^3\LR].
\en
Using (4) and (5) of Lemma \ref{1},
it can be determined that ${\Bbb W}$ satisfies
\eq{symplectic2}
  {\Bbb W}^\ast J {\Bbb W}={\Bbb W} J {\Bbb W} ^\ast =J,
  \en
  where $$J:=\MM1 {0}{0}{-1},\quad
    {\Bbb W}^\ast:=\MM {W_+^\ast} {W_-^\ast}{\bar W_-^\ast}{\bar W_+^\ast}.$$
    This is equivalent to \kak{ccr2}.
Next we show \kak{hccr}.
Note that
$$
[\A _j, b_P(f,l)]=-m\frac{1}{\sqrt2}\lk \frac{e_j^l Q}{\sqrt\omega}, f\rk ,\quad
[\A _j, \bdd(f,l)]=+m\frac{1}{\sqrt2}\lk \frac{e_j^l \bar Q}{\sqrt\omega}, f\rk ,\quad j,l=1,2,3,
$$
and
\eqn
&&
[\hf, b_P(f,l)]=\frac{1}{\sqrt 2}\mmm
\lk
-\hat \A _j\lk \omega^2 T^\ast e_j^l\ww f
 \rk -i\hat{\dot{\A}}_j\lk T^\ast e_j^l\sqrt\omega f\rk \rk,\\
&&
[\hf, \bdd (f,l)]=\frac{1}{\sqrt 2}\mmm
\lk
\hat \A _j\lk \omega^2 \bar T^\ast \tilde e_j^l\ww \tilde f
 \rk -i\hat{\dot{\A}}_j\lk \bar T^\ast \tilde e_j^l\sqrt\omega \tilde f\rk \rk.
 \enn
 Then we have
 \eqn
 &&
 [H_{\pp,P}, b_P(f,l)]\\
 &&=
  -\frac{e}{m} \mmm P_j[\A_j,b_P(f,l)]+\frac{e^2}{m} \mmm \A _j[\A _j,b_P(f,l)]+[\hf,b_P(f,l)]\\
 &&=
  \mmm\lk
 -\frac{e}{m}P_j
 \frac{1}{\sqrt 2}(-m)\lk \frac{e_j^lQ}{\sqrt\omega},
   f\rk +\frac{e^2}{m}
   \A _j(-m)\frac{1}{\sqrt2}
 \lk \frac{e_j^lQ}{\sqrt\omega},  f\rk \right.\\
 &&\hspace{2cm}\left.  +
 \frac{1}{\sqrt 2}\lk
-\hat \A _j\lk \omega^2 T^\ast e_j^l\frac{
  f}{\sqrt\omega}\rk -i\hat{\dot{\A}}_j
  \lk T^\ast e_j^l\sqrt\omega  f\rk \rk\rk \\
&&=
 \frac{1}{\sqrt{2}}\mmm  \lk -\hat \A _j
 \lk T^\ast e_j^l\sqrt\omega \omega f\rk -i\hat{\dot{\A}}_j
 \lk T^\ast e_j^l\frac{1}{\sqrt\omega} \omega f\rk +P_j\lk \frac{e e_j^l Q}{\omega^{3/2}},\omega f\rk \rk\\
&&=
-b_P\lk \omega f,l\rk .
 \enn
 Here we used the fact that $\d \omega^2 T^\ast e_j^l \ww f =T^\ast e_j^l\sqrt\omega \omega f
  -e^2\lk \frac{e_j^l Q}{\sqrt\omega},
   f\rk \vp$. See (2) of Lemma \ref{1}.
Then \kak{hccr} follows.
\qed
\bl{3}There exists a unitary operator $U_P:\ffff_{\pp}\rightarrow \ffff_{\pp}$
such that
\eq{up}
U_P\f \bss_P(f,j)U_P=\ass(f,j),\quad f\in\LR.
\en
\el
\proof
Since $W_-$ is a Hilbert-Schmidt operator on $\oplus^3\LR$,
there exists a canonical linear transformation
$U({\Bbb W})
$ associated with
${\Bbb W}$ \cite{rui}
such that for $F=(F_1,F_2,F_3)\in \oplus^3\LR$, $U({\Bbb W})\f B^\sharp(F) U({\Bbb W})=\ass(F)$, where
$$\vvv {B(F)\\ B^\ast(F)}=\vvv{\iii \jjj \lk a(W_+^{ij}F_j,i)+\add(W_-^{ij}F_j,i)\rk
\\
\iii \jjj \lk a(\bar W_-^{ij}F_j,i)+\add(\bar W_+^{ij}F_j,i)\rk.
}
$$
Since
$$\vvv{b_P(F)\\\bdd(F)}=
\vvv{B(F)\\ B^\ast (F)}+
\vvv{\jjj \sum_{l=1}^3  (P_l L_l^j, F_j)
\\
\jjj\sum_{l=1}^3  (P_l \bar L_l^j,F_j)},$$
we see that
\eq{usp}
U_P:=S(P) U({\Bbb W})
\en
satisfies \kak{up},
where $S(P)$ is the unitary operator given by
\eq{sp2}
S(P):=
 \exp\left(
 \frac{e}{\sqrt2} \mmm \sum_{l=1}^3
  \frac{P_j}{\mass}
  \lk a\lk \frac{e_j^l\vp}{\omega^{3/2}},l\rk-
   \add\lk
   \frac{e_j^l\vp}{\omega^{3/2}},l\rk\rk\right) .
\en
Hence the lemma is complete.
\qed
Let
\eq{upp}
\Omega_{\pp, P}:=U_P\Omega_{\pp}\in \ffff_{\pp},
\en
where $\Omega_{\pp}=\Omega_1 \otimes \Omega_2 \otimes \Omega_3 \in \ffff_{\pp}$.

\bl{unitary}
(1)
It follows that $U_P$ maps $D(\hf^\pp)$ onto $D(H_{\pp,P})(=D(\hf^\pp))$
and
\eq{main22}
U_P\f  H_{\pp,P}U_P=\hff+E_{\pp}(P),
\en
where
\eq{hisp}
E_{\pp}(P)=\frac{1}{2\mass}|P|^2+
\frac{3}{2\pi}
\int_{-\infty}^\infty
\frac{e^2 s^2 \|\vp/(s^2+\omega^2)\|^2}
{m+e^2\|\vp/\sqrt{s^2+\omega^2}\|^2}ds.
\en
(2)
$\Omega_{\pp, P}$ is the unique ground state of $H_{\pp,P}$.
(3)
$\Omega_{\pp, P}$ is the unique vector
such that
$b_P(f,j)\Psi=0$, $j=1,2,3$,  for all $f\in\LR$.
\el
\proof
\kak{hisp} is a minor modification of \cite{hisp1}. Suppose that $b_P(f,j)\Psi=0$ for all $f\in\LR$ and $j=1,2,3$.
Then we have $U_P a(f,j)U_P\f \Psi=0$ and
$U_P\f \Psi=\alpha  \Omega_\pp$, $\alpha\in\CC$.
Hence (3) follows.
By the commutation relation $[H_{\pp,P}, b_P(f,j)]=-b_P(\omega f, j)$ we
can see that
$b_P(f,j) e^{itH_{\pp,P}}\Omega_{\pp, P}=e^{itH_{\pp,P}} b_P(e^{it\omega} f, j)\Omega_{\pp, P}=0$
for all $f\in\LR$. Then
there exists a real number $c$ such that
$e^{itH_{\pp,P}}U_P\Omega_{\pp}=
e^{it c} \Omega_{\pp,P}$
and
$$U_P\f e^{itH_{\pp,P}} U_P\PI \add(f_i,j_i)\Omega_{\pp}
=e^{itc }
\PI \add(e^{it\omega} f_i,j_i)
\Omega_{\pp}.$$
Since the linear hull of
  $\PI\add(f_i,j_i)\Omega_{\pp}$ is dense in $\ffff_{\pp}$,
$$U_P\f e^{itH_{\pp,P}} U_P=
e^{it(\hff+c)}$$ and $c=
E_{\pp}(P)$ follows.
Then (1) is valid. (2) follows from (1).
\qed

\subsection{Total Hamiltonian}
In the previous sections we diagonalized $H_{\pp,P}$ and $H_0$. Thus we can also diagonalize the total Hamiltonian.
Define
\eq{hp}
H_P := H_{\pp,P}\otimes 1 +1\otimes H_0
\en
with domain
\eq{hpdomain}
D(H_P)=D(\hf)
\en
for $P\in\BR$ on $\ffff=\ffff_{\pp}\otimes\ffff_0$.
\bp{eta3}
$H_P$ is $\eta$-self-adjoint.
In particular $H_P$ is closed and
an arbitrary  core of $\hf$ is also a core of $H_P$.
\ep
\proof
The proof is similar to that of Proposition \ref{etaself}.
\qed
We have already shown that
$H_P $ can be diagonalized by making use of the $\eta$-unitary $U_P\otimes V_0$.
We summarize with a proposition.
Let
\eq{psi}
\Psi_P:=\Omega_{\pp, P}\otimes V_0\Omega_0.\en

\bp{dia23}
It follows that
\eq{dia4}
\ov{
(U_P \otimes V_0 )\f
H_P
(U_P \otimes V_0 )\lceil_{D(\hf^\pp)\otimes
\ffff_{0,\rm fin}}}
=
(\hff+E_{\pp}(P))  \otimes 1+1\otimes (\hf^0+E_0).
\en
Moreover
\eq{equa}
H_P \Psi_P=(E_0+E_{\pp}(P))\Psi_P.\en
\ep
\proof
On
$D(\hf^\pp)\otimes \ffff_{0,\rm fin}$ it follows that
\eq{ana2}
(U_P \otimes V_0 )\f
H_P
(U_P \otimes V_0 )
=
(\hff+E_{\pp}(P))  \otimes 1+1\otimes (\hf^0+E_0).
\en
Since
$D(\hf^\pp)\otimes \ffff_{0,\rm fin}$ is a core of the right hand side of
\kak{ana2},  the proposition follows.
\qed

\begin{remark}
The operator $U_P\otimes V_0$ is $\eta$-unitary.
\end{remark}

Next we will indicate the diagonalization of Hamiltonian $H$.
Note that
$$H=\lk
\int_\BR^\oplus H_{\pp,P}dP\rk \otimes 1+1\otimes H_0.$$
 Define the $\eta$-unitary operator by
\eq{uuu}
\mathscr{U}: =\lk \int^\oplus_\BR U_P dp\rk\otimes V_0=
U(-i\vec\nabla) \otimes V_0:\hhh\rightarrow\hhh.
\en
Thus we have the proposition.
\bp{dia22}
$\mathscr{U}$ is $\eta$-unitary on $\hhh$ and
\eq{dia5}
\ov{\mathscr{U}\f H \mathscr U
\lceil_{D(H_\pp)\otimes \ffff_{0,\rm fin}}}
=-\frac{1}{2\mass}
\Delta\otimes 1+1\otimes\hf+
\frac{3}{2\pi}
\int_{-\infty}^\infty
\frac{e^2 s^2 \|\vp/(s^2+\omega^2)\|^2}
{m+e^2\|\vp/\sqrt{s^2+\omega^2}\|^2}ds+E_0.
\en
\ep
\proof
It can be seen that
\eq{ana3}
U(-i\vec\nabla)\f H_{\pp}U(-i\vec\nabla)=\hf^\pp+E_\pp(-i\nabla).
\en
Then by \kak{dia}
\eq{dia55}
\mathscr{U}\f H \mathscr U
=E_\pp(-i\nabla)\otimes 1  +1\otimes\hf+E_0
\en
follows on
$D(H_\pp)\otimes \ffff_{0,\rm fin}$.
Since $D(H_\pp)\otimes \ffff_{0,\rm fin}$ is the core of
the right hand side of \kak{dia55},
the proposition follows.
\qed

\subsection{Heisenberg operators}
In this section we construct
a Heisenberg operator $X(t)$
as a solution to the Heisenberg equation
\eq{sa46}
\frac{d}{dt} X(t)=i[H, X(t)],\qquad X(0)=X,
\en
where we notice that $H$ is not self-adjoint but $\eta$-self-adjoint.
In particular the solution to \kak{sa46}
{\it cannot} always be expressed as $e^{itH}X(0)e^{-itH}$.
So care is required in defining the Heisenberg operator
associated with the non-self-adjoint operator $H$.

Set
\eq{pause}
 \mathscr{H}_\mathscr{S} = \mathscr{S}(\mathbb{R}^3) \hat\otimes \mathscr{F}_\mathscr{S},
 \en
 where $\hat\otimes$ denotes the algebraic tensor product and
	\[ \mathscr{F}_\mathscr{S}
		= {\rm L. H.}\left\{\left. \prod_{i=1}^na^*(f_i,\mu_i)\Omega, \Omega \right|
			f_i \in \mathscr{S}(\mathbb{R}^3), \mu_i=0,1,2,3, i=1, \cdots, n, n\geq 1 \right\}. \]
The dense subspace $\mathscr H_\mathscr{S}$ is useful to study algebraic computations of operators,
since $\mathscr{H}_\mathscr{S} \subset D(H^n)$
for all $n \geq 1$.

\begin{definition}\label{DHO+}{\bf (Heisenberg operators)}
      $X(t)$, $t\in\RR$,  is called the Heisenberg operator associated with $H$ with the initial condition $X(0)=X$
 if and only if
\begin{enumerate}
\item[(1)] For each $t \in \mathbb{R}$,
$X(t)$ is  closed  and
 $\mathscr{H}_\mathscr{S}$ is its core.
\item[(2)] For each $\Psi, \Phi \in \mathscr{H}_\mathscr{S}$,
$H\Phi \in D(X(t))$ and
$( \Psi| X(t) \Phi )$
is differentiable with
	\[ \label{WHE}\frac{d}{dt}( \Psi| X(t) \Phi )
		= i \left(
( H\Psi | X(t) \Phi ) - ( \Psi | X(t) H\Phi ) \right). \]
\item[(3)]
For each $ \Psi,\Phi\in
\mathscr{H}_\mathscr{S}$,
the function
$( \Psi| X(\cdot) \Phi )$ on $\RR$
   can be analytically continued to some domain $\mathscr O\subset\CC$,
which is also denoted by $( \Psi| X(z) \Phi )$ for $z\in\mathscr O$.
Furthermore, for all $n\geq1$,
$H^n\Phi \in D(X)$ and
	\[ \left.
\frac{d^n}{dz^n}( \Psi| X(z) \Phi ) \right |_{z=0}
		= i^n ( \Psi | {\rm ad}^n(H)X \Phi ),
\]
where $( \Psi | {\rm ad}^n(H)X \Phi ) $ is defined by
	\[ ( \Psi | {\rm ad}^n(H)X \Phi )
		= \sum_{j=0}^n \frac{n!}{j!(n-j)!}(-1)^{n-j} ( H^n \Psi | XH^{n-j}\Phi ),
			\quad \Psi, \Phi \in \mathscr{H}_\mathscr{S}.  \]	
\end{enumerate}
\end{definition}
  (2) of Definition \ref{DHO+} is a realization of the Heisenberg equation \kak{sa46} in the weak sense.
(3) ensures the uniqueness of the Heisenberg operator.
See \cite{suz2} for the detail.

Now let us consider the Heisenberg operators
with the initial conditions
$X=p, q, A_\mu(f)$ and $\dot \A_\mu(f)$, where
$p=-i\nabla$ and $q=x$.
Define the operator $\bss(f,j)$ on $\hhh=\LR\otimes\fff$
by $\bss_P(f,j)$ with $P\in\BR$ replaced by
$p$, i.e.,
\eqnn
&&\label{bb1}
 b(f,l):=\frac{1}{\sqrt{2}}
\mmm \lk \hat \A _j\lk
T^\ast e_j^l\sqrt\omega f\rk
+i\hat{\dot{\A}}_j\lk
T^\ast e_j^l\frac{1}{\sqrt\omega} f\rk
-p_j \lk
\frac{ee_j^l Q}{\omega^{3/2}},f\rk
\rk,\quad \quad  \quad \quad \\
&&\label{bb2}
 b^\ast  (f,l):=
\frac{1}{\sqrt{2}} \mmm \lk \hat
\A _j\lk
\bar T^\ast \tilde e_j^l\sqrt\omega \tilde f\rk
-i\hat{\dot{\A}}_j\lk
\bar T^\ast \tilde e_j^l\frac{1}{\sqrt\omega} \tilde f
\rk
-
p_j \lk
\frac{ee_j^l \bar Q}{\omega^{3/2}},f\rk
 \rk. \ennn
Define the operators $\A_j(f,t)$, $\dot{\A}_j(f,t)$, $\A_0(f,t)$, $\dot{\A}_0(f,t)$, $p_j(t)$ and $q_j(t)$, $j=1,2,3$,
by
\begin{eqnarray}
&&
\A _j(f,t)
=
\frac{1}{\sqrt2}\sum_{l=1}^3
\lk
b^\ast \lk
e^{i\omega t}
\ww e_j^l\ov T \hat f ,l\rk
 +b\lk
 e^{-i\omega t}
 \ww e_j^l T\tilde{\hat f},l\rk
 \rk \non \\
 &&\hspace{8cm} -\frac{e}{\mass}
\lk
 \frac {\vp}{\omega^{3/2}},\frac{\hat f}{\sqrt\omega}\rk
 p_j,
 \label{rib1}
 \\
 &&\dot{\A}_j(f,t)
=
\frac{i}{\sqrt2}
\sum_{l=1}^3
\lk
b^\ast \lk
e^{it\omega}
\ww e_j^l\ov T \omega \hat f ,l\rk
-b\lk
e^{-it\omega} \ww e_j^l T \omega \tilde {\hat f},l\rk
\rk,
\label{rib2}\\
&&
\A _0(f,t) = \frac{1}{\sqrt{2}} \lk
\add\lk
e^{it\omega}\frac{1}{\sqrt\omega} \hat f,0\rk
 +
a\lk
e^{-it\omega}\frac{1}{\sqrt\omega} \tilde{\hat f},0\rk
\rk
 \non \\
&&
\label{ztnew}
			 \hspace{4cm}
 - \frac{e}{2}
\lk
\frac{\vp}{\omega^{3/2}},
(e^{it\omega}-1)\frac{\hat f}{\sqrt\omega}+
(e^{-it\omega}-1)\frac{\tilde{\hat f}}{\sqrt\omega}\rk
,
\\
&&
\dot{\A }_0(f,t) = \frac{i}{\sqrt{2}}
\lk
\add\lk
e^{it\omega} \sqrt\omega \hat f,0\rk
-
a\lk
e^{-it\omega} \sqrt\omega \tilde{\hat f},0\rk
\rk\non\\
&&\hspace{6cm}
\label{tzznew}
 -  \frac{ie}{2}\lk
 \frac{\vp}{\sqrt\omega}, e^{it\omega}\hat f-e^{-it\omega}\tilde{\hat f}\rk
 ,\\
&&\label{iraira1}
p_j(t)=p_j,\\
&&\label{iraira2}
q_j(t)=q_j+\frac{t}{m}\lk 1+\frac{e^2}{\mass}\|\vp/\omega\|^2\rk p_j
\non \\
&&\hspace{1cm}+e\frac{i}{\sqrt 2}\sum_{l=1}^3\lkk
b^\ast\lk
(e^{i\omega t}-1)e_j^l\frac{\bar Q}{\omega^{3/2}},l\rk
-
b\lk(e^{-i\omega t}-1)e_j^l
\frac{Q}{\omega^{3/2}},l\rk\rkk.
\end{eqnarray}
\begin{remark}
   All the operators above are defined on $\hhh=\LR\otimes \fff$, but  we omit the tensor notation $\otimes$ for notational  convenience. For example we used $p_j$ for $p_j\otimes 1$ and
$\ass(f)$ for $1\otimes \ass(f)$, etc.
\end{remark}
Since
the operators
$\A_j(f,t)\lceil_{\mathscr H_{\mathscr S}}$,
 $\dot{\A}_j(f,t)\lceil_{\mathscr H_{\mathscr S}}$, $\A_0(f,t)\lceil_{\mathscr H_{\mathscr S}}$, $\dot{\A}_0(f,t)\lceil_{\mathscr H_{\mathscr S}}$, $p_j(t)\lceil_{\mathscr H_{\mathscr S}}$ and $q_j(t)\lceil_{\mathscr H_{\mathscr S}}$ are closable, we denote their closed extensions simply by
 $\A_j(f,t)$,
 $\dot{\A}_j(f,t)$, $\A_0(f,t)$, $\dot{\A}_0(f,t)$, $p_j(t)$ and $q_j(t)$, respectively.

\bt{h-o}
Let $f\in C_0^\infty(\BR)$. Then $\A_\mu(f,t)$ (resp. $\dot\A_\mu(f,t)$, $p(t)$,
 $q(t)$) is the Heisenberg operator associated with $H$
with the initial condition
 $\A_\mu(f,0)=\A_\mu(f)$ (resp. $\dot\A_\mu(f,0)=\dot\A_\mu(f)$, $p(0)=p$, $q(0)=q$).
\et
The heuristic idea of the proof of Theorem \ref{h-o} is as follows.
We note that
$\A_0$ commutes with $H_{\pp}$
and
$\A_j$,  $p$, $q$
commute with $H_0 $.
So
the informal solutions to the Heisenberg equation
\kak{sa46} for the initial condition $X=q, p, \A_j(f)$ and $\dot \A_j(f)$ are  given by
 \eq{tp}
\widetilde q_j(t):=   e^{itH_{\pp}}q_je^{-itH_{\pp}}, \quad
   \widetilde p_j(t):=e^{itH_{\pp}}p_je^{-itH_{\pp}}
  \en
  and
	\eq{ta}
  \widetilde  \A_j(f,t):= e^{itH_{\pp}}\A_j (f)e^{-itH_{\pp}}, \quad  \widetilde {\dot {\A}}_j(f,t):=e^{itH_{\pp}}\dot{\A }_j(f)e^{-itH_{\pp}}
 \en
 for $j=1,2,3$, respectively.

Moreover since $\A_0(f)$ and $H_\pp$ commute,
in order to construct
 the Heisenberg operators with initial conditions
 $\A_0(f)$ and $\dot \A_0(f)$,
it is enough to find  the Heisenberg operators
$\widetilde \A_0(f,t)$ and
$\widetilde{\dot \A}_0(f,t)$
associated with $H_0$ instead of $H$:
\eq{ohei2}
\frac{d}{dt}  \widetilde \A_0(f,t)=i[H_0,\widetilde \A_0(f,t)],\quad
\frac{d}{dt}  \widetilde{\dot \A}_0(f,t)=i[H_0,
\widetilde{\dot \A}_0(f,t)].
\en
We will show that $\widetilde \A_\mu(f,t)=\A_\mu(f,t)$,
$\widetilde {\dot \A}_\mu(f,t)=\dot{\A}_\mu(f,t)$, $\widetilde p_j(t)=p_j(t)$ and $\widetilde q_j(t)=q_j(t)$ on $\hhh_{\mathscr S}$
and prove that they are the Heisenberg operators.

{\it Proof of Theorem \ref{h-o}}

By the assumption $f\in C_0^\infty(\BR)$ and
(2) of Assumption \ref{ass2},
$\|e^{+\epsilon\omega} \vp\|_\infty<\infty$,
it is immediate
that
$\ab{\A_\mu(f,t)}$ (resp. $\ab{\dot\A_\mu(f,t)}$, $\ab{p_j(t)}$ and $\ab{q_j(t)}$) can be analytically
continued to some domain with respect to $t$.
So it is enough to check (2) of Definition \ref{DHO+}.

We see directly that
\kak{ztnew} and \kak{tzznew}
satisfy
the Heisenberg equation
\kak{sa46}.

Next we examine \kak{rib1} and \kak{rib2}.
The vector potentials
$\A _j(f)$ and $\dot{\A}_j(f)$
can be  expressed by means of $\bdd$ and $b_P$.
In fact direct computation shows that
\eqnn
\A _j(f)
&=&
\frac{1}{\sqrt2}
\sum_{l=1}^3
\lk
\bdd\lk
\ww e_j^l\ov T \hat f ,l\rk
+
b_P\lk
\ww e_j^l T\tilde{\hat f},l\rk
\rk -eP_j
\lk
\frac{\vp}{{\mass}\omega^{3/2}},
\frac{\hat f}{\sqrt\omega}\rk,\non\\
&&\label
{ta2} \\
\dot{\A}_j(f)&=&
\frac{i}{\sqrt2}\sum_{l=1}^3
\lk
\bdd\lk
\ww e_j^l\ov T \omega \hat f ,l\rk
-b_P\lk
\ww e_j^l T\omega \tilde{\hat f},l\rk
\rk.
\ennn
Note that
$$e^{itH_{\pp,P}}b_P(f,j) e^{-itH_{\pp,P}}
=b_P(e^{-i\omega t}f,j),\quad
e^{itH_{\pp,P}}\bdd (f,j) e^{-itH_{\pp,P}}
=\bdd (e^{i\omega t}f,j).$$
Then
$\A_j(f,t,P)=e^{itH_{\pp,P}}\A_j(f) e^{-itH_{\pp,P}}$ is given by
\begin{eqnarray}
&&
\A _j(f,t,P)
=
\frac{1}{\sqrt2}\sum_{l=1}^3
\lk
b^\ast \left(e^{i\omega t}
\ww e_j^l\ov T \hat f ,l\right)
 +b\left(
 e^{-i\omega t}
 \ww e_j^l T\tilde{\hat f},l\right)\rk \non\\
 &&\label{kok1}
\hspace{5cm}-e
 P_j
 \left(
 \frac{\vp}{{\mass}\omega^{3/2}},\frac{\hat f}{\sqrt\omega}\right).
 \end{eqnarray}
Thus \kak{rib1}
   satisfies the Heisenberg equation \kak{sa46}.

Finally we analyze  \kak{iraira1} and \kak{iraira2}.
By the dipole approximation, $p$ and $e^{itH_\pp}$ commute. Then $e^{itH_\pp}p_je^{-itH_\pp}=p_j$.
Thus it is trivial to see that $p_j$ is the Heisenberg operator.
We see that
\eqn
q_j(t)\Psi
&=&  \int_0^t i e^{isH_\pp}[H_\pp,q_j]e^{-isH_\pp}\Psi ds+q_j\Psi\\
&=&\frac{1}{m}
\int^t_0
e^{isH_\pp}(p_j-e\A_j)e^{-isH_\pp}\Psi ds+q_j\Psi\\
&=&
\frac{t}{m}p_j\Psi +q_j\Psi-\frac{e}{m}\int_0^t \A_j(\varphi,s)\Psi ds.
\enn
By \kak{rib1} we can compute
$\frac{e}{m}\int_0^t \A_j(\varphi,s)\Psi ds$ as
\eqn
&&\frac{e}{m}\int_0^t \A_j(\varphi,s)\Psi ds\\
&&=
\frac{e}{m}
\frac{i}{\sqrt 2}\sum_{l=1}^3
\lkk
b^\ast\lk
(e^{i\omega t}-1)e_j^l\frac{\bar Q}{\omega^{3/2}},l\rk
-
b\lk
(e^{-i\omega t}-1)e_j^l
\frac{Q}{\omega^{3/2}},l\rk\rkk
+\frac{t}{m} \frac{e^2}{\mass}\|\vp/\omega\|^2 p_j.
\enn
Then \kak{iraira2} satisfies the Heisenberg equation \kak{sa46}.
\qed

We utilize \kak{tp}-\kak{ohei2}, 
Maxwell's equations and Newton's equation of motion for NRQED.
For all $f \in \mathscr{S}(\mathbb{R}^3)$,
	\begin{align*}
		& \frac{d^2}{dt^2}{\vA }(f,t) - \vA (\Delta f,t)	 
			= \int_{\mathbb{R}^3}\vec J(x,t)f(x)dx ,  \\
		& \frac{d^2}{dt^2}{\vA }_{0}(f,t) - \vA _{0} (\Delta f,t)
			= \int_{\mathbb{R}^3}\rho(x,t)f(x)dx,
	\end{align*}
and
	\[ \frac{d^2}{dt^2}q(t) = -e \dot{\vA }(\varphi,t), \]
where
\eqn
&&
\vec J(x,t) = \frac{e}{2m}\lk
 \varphi(x)(p(t)-e\vA (\varphi,t))+
(p(t)-e\vA (\varphi,t)) \varphi(x)\rk,\\
&&
\rho(x,t) =
e\varphi(x). \enn

\section{LSZ formalism and asymptotic completeness}
 We shall
 construct the asymptotic field
 $\add_{P,\pm}(f,\mu)$ by the LSZ method in this section.
Let
	\begin{eqnarray}
		&& a_{P,t}(f,j) := i\mmml
(\dot{\A}_l(f_t^{l,j},t,P)-
\A _l( \dot f_t^{l,j},t,P)),\quad j=1,2,3, \\
		&& a_{P,t}(f, 0) := i(\dot{\A}_{0}
(f_t^0, t)-{\A}_{0}(\dot f_t^0, t)),
	\end{eqnarray}
for $f \in \LR$,
where both $\A_0(f,t)$ and $\dot \A_0(f,t)$ are regarded as operators in $\ffff$,
$\A _l( f,t,P)=
e^{itH_{\pp,P}}
\A _l( f)e^{-itH_{\pp,P}}$,
 $\dot \A _l( f,t,P)=
 e^{itH_{\pp,P}}
\dot \A _l( f)e^{-itH_{\pp,P}}$
and
\eqnn
\label{wak}
&&
{f}_t^0 =
F\f\lk
\frac{e^{it\omega}}{\sqrt{2\omega}}\tilde{f}\rk,
\quad
 f_t^{l,j}(k)= 
 F\f\lk
 \frac{e^{+it\omega}}{\sqrt{2\omega}}\tilde e_l^j \tilde{f}\rk,\quad j=1,2,3,\\
&&\dot{f}_t^0 =
F\f\lk
i\omega \frac{e^{it\omega}}{\sqrt{2\omega}}\tilde{f}\rk,
\quad
 \dot f_t^{l,j}(k)= F\f
 \lk
 i\omega \frac{e^{+it\omega}}{\sqrt{2\omega}}\tilde e_l^j \tilde{f}\rk,
\label{wakat}
\ennn
and
$F\f$ denotes the inverse Fourier transformation of $\LR$.
We also set
	\[ a^{\dagger}_{P,t}(f, \mu) = \lkk\begin{array}{ll}
\lk a_{P,t}(\bar f, \mu)\rk ^\ast,&\mu=j=1,2,3,\\
-\lk a_{P,t}(\bar f,0)\rk^\ast, & \mu=0.
\end{array}
\right.
\]
From the expression of  \kak{asss1} and \kak{asss2} it can be seen that
\eq{ath}
a_{P,t}(h,j)=
e^{itH_{\pp,P}}e^{-it\hff}a (h,j)
e^{it\hff}e^{-itH_{\pp,P}}
\en
for $j=1,2,3$ and
\eq{athz}
a_{P,t}(f,0)=a(f,0)-\frac{e}{\sqrt2}
\lk
\frac{\vp}{\omega^{3/2}}, (1-e^{it\omega})f
\rk
.
\en
From \kak{athz}, the strong limit of
$\ass_{P,t}(f,0)$ as $t\rightarrow \pm\infty$
is easily obtained.
In order to have an explicit form for
$a_{P,t}^\sharp (h,j)$, $j=1,2,3$,
it is enough to obtain
explicit forms for
$\A_l(f,t)$ and $\dot\A_l(f,t)$.
Fortunately this can be done using \kak{rib1} and \kak{rib2}.
In the next lemma we show that
the strong limits of $\ass_{P,t}(f,\mu)$
can be represented by
$b_P^\sharp(f,\mu)$ defined in
\kak{b100},\kak{b200}, \kak{b1} and \kak{b2}.

\bl{sca}
Let $\Psi \in \ffff_{\rm fin}$.
Then the strong limits
\eqnn
&&
a_{P,{\rm out/in}}(h,\mu)\Psi=s-\lim_{t\rightarrow \pm\infty}a_{P,t}(h,\mu)\Psi,\\
&&
a_{P,{\rm out/in}}^\dagger(h,\mu)\Psi=s-\lim_{t\rightarrow \pm\infty}a_{P,t}^\dagger(h,\mu)\Psi
\ennn
exist where $"out"$,  $"in"$  stand for $t\rightarrow +\infty$, $t\rightarrow -\infty$ respectively,
  and are given explicitly by
\eqnn
&& \label{s1}
a_{P,{\rm in}}(h,j)=b_P(h,j),\\
&&\label{ss1}
a_{P,{\rm in}}^\dagger(h,j)=b_P^\ast(h,j),\\
&& \label{s2}
a_{P,{\rm out}}(h, j) = \sum_{i=1}^3b_P(L^{ij}h,i),\\
&&\label{s3}
a^\dagger_{P,{\rm out}}(h, j) = \sum_{i=1}^3b_P^\ast (\bar{L}^{ij}h,i),\\
&&
\label{zeroo}
 a_{P,{\rm out}}(h,0)= b(h,0)=a_{P,{\rm in}}(h,0),\\
&&
\label{zero}
 a_{P,{\rm out}}^\dagger  (h,0) = b^\dagger(h,0)=a_{P,{\rm in}}^\dagger (h,0),
 \ennn
 where
$ L^{ij}h= \delta_{ij}h-i\pi e^2Q\hat{\varphi}\omega\mmml e^i_l[e^j_lh] $
and
$[f](k):=\int_{S_2} f(|k|S )dS$, $dS=\sin \theta d\theta d\phi$, $0\leq \theta\leq 2\pi$, $0\leq \phi\leq \pi$.
\el
\proof
The proof is parallel with \cite{a83b}.
\kak{zeroo} and \kak{zero}
can be proven by the Riemann-Lebesgue lemma.
We shall prove \kak{s1}-\kak{s3}.
By Theorem \ref{h-o}  we have
\eqn
&&
i\mmml(\dot{\A}_l(h_t^{l,j},t,P)-\A _l( \dot h_t^{l,j},t,P))\\
&&=
\iii
\lk
\bdd
(e^{-it\omega}(\bar W_-^{ji})^\ast
e^{+it\omega}\hat h ,i)
+
b_P(e^{-it\omega} (W_+^{ji})^\ast  e^{+it\omega}
\hat h,i)\rk +const.
\enn
Since we can see that $W_-^{ji}$ is an integral operator with kernel in $L^2(\BR\times \BR)$, it is a Hilbert-Schmidt operator. Then
$\|(\bar W_-^{ji})^\ast
e^{it\omega}h\|\rightarrow 0$ as $t\rightarrow \pm\infty$. Hence
$\bdd
(e^{-it\omega}(\bar W_-^{ji})^\ast
e^{it\omega}\hat h ,i)
\rightarrow 0$ as $t\rightarrow \pm\infty$.
Next we shall estimate
$b_P(e^{-it\omega} (W_+^{ji})^\ast  e^{it\omega}
\hat h,i)$.
Let $X_{ij}(t)=e^{-it\omega} (W_+^{ji})^\ast  e^{it\omega}
h$. Then
\eqn
\frac{d}{dt}X_{ij}(t)
&=&
\frac{-i}{2}
\mmml
e^{-it\omega} e^i_l\left[
\omega, \ww T\w+\w T\ww\right]
e_l^j
e^{it\omega}
h\\
&=&
\frac{-i}{2}
\mmml
e^{-it\omega} e^i_l[\omega^2, T]
e_l^j
e^{it\omega}
h
\\
&=&
\frac{-ie^2}{2}
\mmml
 \frac{e^{-it\omega} e^i_l Q}{\w}\lk
\frac{e_l^j e^{-it\omega}\vp}{\w},h\rk:=
\frac{-ie^2}{2}\varrho_{ij}(t,\cdot).
\enn
 Then
\eq{int156}
X_{ij}(t)=(W_+^{ji})^\ast h
+\frac{-ie^2}{2}
 \int_0^t ds \varrho_{ij}(s,\cdot).
\en
Since
$\d \left|\lk
\frac{e_l^j  e^{-is\omega}\vp}{\w},h\rk\right|\leq {\rm const.}/s^2$,
the integral of the right-hand side of \kak{int156} as $t\rightarrow \pm\infty$ is well defined.
First we investigate the case $t\rightarrow -\infty$.
Then
\eqn
&&
\frac{-ie^2}{2}
 \int_0^{-\infty}  ds
\varrho_{ij}(s,k)\\
&&
=
\frac{ie^2}{2}
\mmml  \lim_{\epsilon\downarrow 0}
 \int_{-\infty}^0 ds\int dk'
 e^{-is(\omega(k)-\omega(k')+i\epsilon)}
\frac{e_l^i(k)e_l^j(k')Q(k)\vp(k')h(k')}{\sqrt{\omega(k)}\sqrt{\omega(k')}}
\\
&&=
-\frac{e^2}{2}
\mmml
\lim_{\epsilon\downarrow 0}
\int dk'
\frac{e_l^i(k)e_l^j(k')Q(k)\vp(k')h(k')}
{(\omega(k)-\omega(k')+i\epsilon)
\sqrt{\omega(k)}\sqrt{\omega(k')}}\\
&&=
-\frac{e^2}{2}
\mmml
\lim_{\epsilon\downarrow 0} \int dk'
\frac{(\omega(k)+\omega(k'))
e_l^i(k)e_l^j(k')Q(k)\vp(k')h(k')}
{(\omega(k)^2-\omega(k')^2+i\epsilon)
\sqrt{\omega(k)}\sqrt{\omega(k')}}\\
&&=-(W_+^{ji})^\ast h(k) +\delta_{ij} h(k).
\enn
 Hence
$\lim_{t\rightarrow -\infty} X_{ij}(t)h=\delta_{ij}h$
and \kak{s1}, i.e., $a_{{\rm in}}(h,j)=b_P(h,j)$ follows.
Next we show that
\eq{s4}
\lim_{t\rightarrow +\infty} X_{ij}(t)h=
\delta_{ij}h-ie^2\pi Q\vp\omega e_\mu^i[he_j^l]
.
\en
We have
\eqn
\lim_{t\rightarrow +\infty} X_{ij}(t)h=(W_+^{ji})^\ast h+\frac{-ie^2}{2}
 \int_0^\infty  ds
\varrho_{ij}(s,\cdot)
=\frac{-ie^2}{2}\int_{-\infty}^\infty ds
\varrho_{ij}(s,\cdot)+
\delta_{ij}h.
 \enn
Since, by the Fourier transformation, we have
$$ \frac{-ie^2}{2}\mmml
\int_{-\infty}^\infty
 \frac{e^{-is\omega} e^i_l Q}{\w}\lk
\frac{e_l^j e^{-is\omega}\vp}{\w},h\rk ds
=-ie^2\pi Q\vp\omega \mmml e_l^i[he_l^j],$$
\kak{s4}
and then
\kak{s2} follows.
\kak{s3} is similarly proven.
Then the proof is complete.
\qed


In what follows "$\oi$" stands for "out" or "in".
Next we consider
the asymptotic field $\ffff_P^\oi $
and construct the scattering operator $S$ connecting $\ffff_P^{{\rm in}}$ and
$\ffff_P^{{\rm out}}$.
We denote by $\ffff _{P,{\rm fin}}^\oi =
\ffff _{P,{\rm fin}}^\oi
$
the linear hull of the set
$$
 \lkk\skima  \PI a_{P,\oi} ^{\dagger}(h_i, \mu_i)  \Psi_P, \Psi_P
		\left|
 h_i \in \mathscr{S}(\mathbb{R}^3), \mu_i = 0,1,2,3,~ i=1, \cdots n,~ n\geq 1\skima \right. \rkk
$$
and by $\ffff_P ^\oi $ the closure of
$\ffff _{P,{\rm fin}}^\oi $ in $\ffff$.
In the next lemma, commutation relations are established.
\begin{lemma} \label{suzuki}
The following commutation relations hold for
$\ffff _{P,{\rm fin}}^\oi $:
	\begin{eqnarray}
&&\label{l3}
 [a_{P,\oi} (h,\mu), a_{P,\oi} ^{\dagger}(g,\nu)] = -g_{\mu\nu}(\bar h, g),\\
&&\label{l3.5}
 [a_{P,\oi} (h,\mu), a_{P,\oi}(g,\nu)] = 0 = [a_{P,\oi}^\dagger (h,\mu), a_{P,\oi}^\dagger(g,\nu)],\\
&&\label{l14}
 [H_P , a_{P,\oi} (h,\mu)] = -a_{P,\oi} (\omega h, \mu), \\
&&\label{l15}
 [H_P , a^{\dagger}_{P,\oi} (h,\mu)] = a^{\dagger}_{P,\oi} (\omega h, \mu)
	\end{eqnarray}and
$
a_{P,\oi} (h,\mu)\Psi_P=0
$
for all $h\in\LR$.
\end{lemma}
\proof
\kak{l14} and \kak{l15} follow directly from the commutation relations between $H_P $ and $b^\sharp$.
The commutation relations in \kak{l3} for $\mu=0$ or $\nu=0$ are obtained by direct computation.
Other commutation relations can be proven by
$a_{P,\oi}(h,j):=\lim_{t\rightarrow\pm\infty}
a_{P,t}(h,j)=
e^{itH_{\pp,P}}e^{-it\hff}a (h,j)
e^{it\hff}e^{-itH_{\pp,P}}$
and a limiting argument.
\qed

We constructed the quadruple
\eq{tr1}
(\ffff_P^\oi , H_P , \{a_{P,\oi} (h,\mu), a^{\dagger}_{P,\oi} (h,\mu) | h \in L^2(\mathbb{R}^3) \}, \Psi_P)
\en
relevant to \kak{omega},  including the scalar potential.

\bt{asy}
{\bf (Asymptotic completeness)}
It follows that
$ \ffff_P ^{\rm in} = \ffff_P ^{{\rm out}}
= \ffff .
$
\et
\proof
Let
$$
\ffff ^\oi _{P,{\rm fin, TL}}=
		\lkk\skima  \PI a_{P,\oi} ^{\dagger}
(h_i, j_i) \Omega_{\pp, P}, \Omega_{\pp, P}
\left|
 h_i \in \mathscr{S}(\mathbb{R}^3), j_i = 1,2,3,~ i=1, \cdots n,~ n\geq 1\skima \right. \rkk
$$
and
$$
\ffff ^\oi _{{\rm fin}, 0}=		
		\lkk\skima  \PI
a_{P,\oi} ^{\dagger}(h_i, 0)
 V_0\Omega_0, V_0\Omega_0
	\left|
 h_i \in \mathscr{S}(\mathbb{R}^3),~ i=1, \cdots n,~ n\geq 1\skima \right. \rkk.
$$
Since
$\ffff ^\oi _{P,{\rm fin}}
		= \ffff ^\oi _{P,{\rm fin, TL}} \hat \otimes \ffff ^\oi _{\rm fin, 0}$,
we need only prove that $\ffff ^\oi _{P,{\rm fin, TL}}$ (resp. $\ffff ^\oi _{{\rm fin}, 0}$)
is dense in $\ffff _{\pp}$ (resp. $\ffff _0$).
We assume that there exists a vector $\Phi  \in \ffff _{\pp}$ such that
	\[ \lk
\PI a^\dagger_{P,{\rm in}}(h_i,j_i)
  \Omega_{\pp, P}, \Phi \rk=0 \]
for all $h_i$ and $j_i=1,2,3$.
By Lemma \ref{sca} and the relations $U_P\f b^\sharp(f,j)U_P=\ass(f,j)$,
we have
$$\lk
\PI a^\dagger(h_i,j_i)
\Omega_{\pp},U_P^{-1}\Phi \rk
=0$$
for all $h_i$ and $j_i=1,2,3$.
Thus $\Phi =0$, which yields that $\ffff^{\rm in}_{P,{\rm fin, TL}}$ is dense in $\ffff_{\pp}$.
Similarly, suppose that
$(\PI a^\dagger_{P,{\rm out}}(h_i,j_i)
 \Omega_{\pp, P},\Phi )=0$ for all $h_i$ and $j_i=1,2,3$.
Then we have
\eq{sp}
\sum_{i_1,...,i_n=1}^ 3
\lk
\PI a^\dagger(\bar L^{l_ij_i}h_i,l_i)
 \Omega_{\pp}, U_P^{-1} \Phi \rk
=0.
\en
Let $L=\lk L^{ij}\rk _{1\leq i,j\leq 3}\oplus^3\LR
\rightarrow \oplus^3\LR$.
We note that, as a consequence, $L=\lim_{t\rightarrow+\infty} e^{-it\omega}
W_+^\ast e^{it\omega}$.
From the symplectic structure
${\Bbb W}^\ast J {\Bbb W}={\Bbb W} J {\Bbb W} ^\ast =J$, it follows
that $W_+^\ast W_+-W_-^\ast W_-=1$. In particular it follows that
$e^{-it\omega}W_+^\ast W_+ e^{it\omega}=e^{-it\omega} W_-^\ast W_-e^{it\omega}+1$.
Thus
$$L L^\ast =\lim_{t\rightarrow +\infty}e^{-it\omega}
W_+^\ast W_+ e^{it\omega}=1,$$
since $e^{-it\omega} W_-^\ast W_-e^{it\omega}$ vanishes
as $t\rightarrow \pm\infty$.
Then $L$ has an inverse as an
operator from $\oplus^3\LR$ to itself and the linear hull of
vectors
of the form
$\PI a^\dagger(Lf_i)
 \Omega_{\pp}$ is dense in $\ffff_{\pp}$.
Hence \kak{sp} implies that $\Phi =0$ and $\ffff^{\rm out}_{P,{\rm fin, TL}}$ is dense in $\ffff_{\pp}$.

We prove that $\ffff _{{\rm fin},0}^\oi $ is dense in $\ffff _{0}$.
Denoting by $\ffff _{{\rm fin}, 0}$ the linear hull of the set
$$
		\lkk\skima  \PI b^{\dagger}(h_i, 0)
 V_0\Omega_0, V_0\Omega_0
\Big |
 h_i \in \mathscr{S}(\mathbb{R}^3),~ i=1, \cdots n,~ n\geq 1\skima \rkk,
$$
by Lemma \ref{sca}, we have $\ffff _{{\rm fin}, 0} = \ffff ^{\oi}_{{\rm fin}, 0}$,
and hence we need only prove that $\ffff _{{\rm fin},0}$ is dense in $\ffff _0$.
Setting
$$
		\mathcal{D}_0 =
		\lkk\skima
\PI \addd(h_i, 0)
 \Omega_0, \Omega_0
		\Big | h_i \in \mathscr{S}(\mathbb{R}^3),~ i=1, \cdots n,~ n\geq 1\skima \rkk,
	$$
we have the result that the linear hull of $\mathcal{D}_0$ is dense in $\ffff _0$.
Let
	\[ U_0 = {\rm exp}\left(
 \frac{e}{\sqrt{2}}
 \left(a^\ast \left(\frac{\hat{\varphi}}{\omega^{3/2}},0\right)
-a\left(\frac{\hat{\varphi}}{\omega^{3/2}},0\right) \right) \right). \]
Then we observe that $U_0$ is unitary and that
	\begin{equation}
	\label{UV}
		V_0\Omega_0= e^{e^2/2\|\hat\varphi/\omega^{3/2}\|^2}U_0\Omega_0.
	\end{equation}
We shall prove $\mathcal{D}_0 \subset U_0^{-1}\ffff _{{\rm fin}, 0}$ by induction.
It is clear from \eqref{UV} that $\Omega_0 \in U_0^{-1}\ffff _{{\rm fin}, 0}$.
Assume that
$
		\PI \addd(h_i,0)
 \Omega_0 \in U_0^{-1}\ffff _{{\rm fin}, 0}$.
Then we have
\eqn
&&
\prod_{i=1}^{n+1}  \addd(h_i,0)
\Omega_0\\
&&= \lk\skima 
b^\dagger(h_{n+1},0) - e
(\vp/\omega^{3/2},  h_{n+1}) \rk
\PI \addd(h_i,0)
\Omega_0
+ e(\vp/\omega^{3/2}, h_{n+1})
\PI \addd(h_i,0)
\Omega_0  \\
		&&=U_0^{-1}b^\dagger(h_{n+1},0) U_0
\PI \addd(h_i,0)
\Omega_0
		+ e (\vp/\omega^{3/2},  h_{n+1})
\PI \addd(h_i,0)
\Omega_0.
	\enn
It follows that $\prod_{i=1}^{n+1} \addd(h_i,0)
 \Omega_0 \in U_0^{-1}\ffff _{{\rm fin}, 0}$
and we have the desired result.
\qed
Let
   $S_P:\ffff_P^{{\rm out}}\rightarrow \ffff_P^{\rm in}$
   be defined by
\eq{scattering}
S_P\PI a_{P,{\rm out}}
^\dagger(f_i,\mu_i)
\Psi_P:=
\PI a_{P,{\rm in}}
^\dagger(f_i,\mu_i)
\Psi_P.
\en
Then $\|S_P\Phi\|=\|\Phi\|$ for $\Phi\in \ffff_{\rm fin}^{{\rm out}}$ follows from
\eqref{zeroo} and \eqref{zero} in Lemma \ref{sca}
and the commutation relations \eqref{l3} and \eqref{l3.5} in Lemma \ref{suzuki}.
Thus $S_P$ can be extended to a unitary operator
from $\ffff_P^{{\rm out}}$ to
$\ffff_P^{\rm in}$.
$S_P$ is called the scattering operator.
\begin{theorem}
   $S_P$ is unitary and $\eta$-unitary, i.e.,
$S_P^\ast =S_P^{-1}=S_P^{\dagger}$.
\end{theorem}
\begin{proof}
The unitarity of $S_P$ is already proven.
  $[S_P, \eta] = 0$ implies that $S_P$ is $\eta$-unitary.
\end{proof}

\section{Physical subspace}
\subsection{Abstract setting}
We begin with an abstract version of physical subspace.
Let $\mathcal{K}$ be a Klein space with a metric $( \cdot | \cdot)$.
For a densely defined linear operator $X$ on $\mathcal{K}$,
we denote by $X^{\dagger}$ the adjoint of $X$ with respect to
$( \cdot | \cdot)$.
We denote the set of densely defined operators on $\mathcal{K}$ by $\mathscr{C}(\mathcal {K})$.
\begin{definition}
\label{dis222}
The map $F:\mathscr{S}(\mathbb{R}^3)\rightarrow
\mathscr{C}(\mathcal {K})$ is called an operator valued distribution if and only if there exists
a dense subspace $\mathcal{D}$ such that
\begin{itemize}
\item[(1)]
 $F(\alpha f+\beta g)\Psi=(\alpha F(f)+\beta F(g))\Psi$ for $\alpha,\beta\in\CC$, $f,g\in \mathscr{S}(\mathbb{R}^3)$
and $\Psi \in \mathcal{D}$;
\item[(2)]
 the map
$\mathscr{S}(\mathbb{R}^3) \ni f \mapsto ( \Psi |F(f) \Phi)$ is a tempered distribution for $\Psi, \Phi \in \mathcal{D}$.
\end{itemize}
\end{definition}

\begin{definition}\label{dis}
Let $B=\{B(\cdot,t)\}_{t\in\RR}$  be a family of operator valued distributions.
 This family is in class $\mathscr D(\mathcal{K})$ if and only if
\begin{itemize}
\item[(1)]
 there exists a dense subspace $\mathcal{D}_B$ in $\mathcal{K}$ such that,
for all  $t \in \mathbb{R}$ and $f \in \mathscr{S}(\mathbb{R}^3)$,
$\mathcal{D}_B \subset D(B(f,t)) \cap D(B(f,t)^{\dagger})$ and
$B(f,t)^{\dagger}|_{\mathcal{D}_B}=
B(\bar{f},t)|_{\mathcal{D}_B}$;
\item[(2)] for each $\Psi \in \mathcal{D}_B$,
$B(f,t)\Psi$ is strongly differentiable
in $t$
and its derivative, denoted by $\dot{B}(f,t)\Psi$, is continuous in $t$.
\end{itemize}
\end{definition}
By Definition \ref{dis},
        $\{B(\cdot, t)\}_{t\in\RR}\in \mathscr{D}(\mathcal {K})$ implies that
 $\dot{B}(\cdot,t)$ is also an operator-valued distribution which satisfies (1) of Definition \ref{dis} with $B$ replaced by $\dot B$.
We now provide an abstract definition of a free field .
\begin{definition}
\label{free}
A family of operator valued distributions $\{B(\cdot,t)\}_{t\in\RR}\in\mathscr{D}(\mathcal{K})$ is called a free field if and only if
$B(f,t)\Psi$ is strongly two-times differentiable in $t$ and
\eq{freef2}
\frac{d^2}{dt^2}B(f,t)\Psi-B(\Delta f,t)\Psi=0
\en
holds for all $f\in \mathscr{S}(\BR)$ and
$\Psi\in \mathcal {D}_B$.
The set of free fields is denoted by $\df$.
 \end{definition}

Further to introducing the Gupta-Bleuler subsidiary condition,
the positive frequency part of it has to be defined.
Let $\{B(\cdot,t)\}_{t\in\RR}\in\mathscr{D}(\mathcal{K})$.
Then one can automatically
construct a free field from
$B(\cdot,t)$, as described below.
Define
\eq{dis2}
 c_s(g) :=
 i\left(\dot{B}(g_s,s)-B(\dot g_s,s) \right),
 \en
where  $g_s$ and $\dot g_s$ are defined by
	\begin{equation}
	\label{gt}
		{g}_s = F\f\left(
\frac{\tilde g}{2\omega}e^{is\omega}
\right),\quad
\dot g_s=\partial_s g_s=F\f
\left(
i
\frac{\tilde g}{2}e^{is\omega}
\right)
.
	\end{equation}
Note that in \kak{dis2}
$$\dot B(g_s,s)=\dot B(f,s)\lceil_{f=g_s}.$$
Set
$  c_s^{\dagger}(h) := \lk
c_s(\bar{h})\rk ^{\dagger}$.
Let us define the
 operator $F(f,s,t):\mathcal {K}\rightarrow \mathcal {K}$ for $f\in \mathscr{S}(\BR)$ and $s,t\in\RR$,
	by
\begin{equation}
	\label{F}
		F(f,s,t) := c_s\left(e^{-it\omega}\tilde{\hat{f}}\right) +
c_s^{\dagger}\left(
e^{it\omega}\hat{f}\right).
	\end{equation}
It can be proven  that for each $s\in\RR$,
$\{F(\cdot, s, t)\}_{t\in\RR}\in \df$.
Then we can define
the family of maps $\Theta_s$, $s\in\RR$,
$$\Theta_s:\mathscr{D}(\mathcal{K})\rightarrow \df,\quad
\{B(\cdot,t)\}_{t\in\RR}
\mapsto\{F(\cdot,s,t)\}_{t\in\RR}.$$
In particular $\Theta_s$ leaves $\df$ invariant.
In the next lemma a stronger statement is established.
\begin{lemma} \label{positivepart}
Let $B=\{B(\cdot,t)\}_{t\in\RR}\in \mathscr {D}$.
Then
\begin{enumerate}
\item[(1)] $B(f, t) = F(f,t,t)$ holds for all
$f\in\mathscr{S}(\BR)$ and $t\in\RR$;
\item[(2)] If, in addition, we assume that   $B\in \df$
and that for each $\Psi\in \mathscr{D}_B$, there exists a continuous semi-norm $C_{\Psi}$
on $\mathscr{S}(\mathbb{R}^3)$ such that
\eq{suzuki2}
 \sup_{t \in \mathbb{R}}\|B(f,t)\Psi\|+ \sup_{t \in \mathbb{R}}\|\dot B(f,t)\Psi\| \leq C_{\Psi}(f),
 \en
then  $c_s(h)$ (resp. $c^{\dagger}_s(h)$) is
  independent of $s \in \mathbb{R}$ and
		\[ B(f,t) = c(e^{-it\omega}\tilde{\hat{f}}) + c^{\dagger}(e^{it\omega}\hat{f}). \]
holds. Here we set $c_s=c$ and $c_s^\dagger=c^\dagger$.
\end{enumerate}
\end{lemma}

\proof
We have
	\begin{align*}
		& c_s(e^{-it\omega}\tilde{\hat{f}},t)
			= i\left( \dot{B}\left(\frac{e^{-i(t-s)\omega}}{2\omega}f, s\right)
				- iB
\left(\frac{e^{-i(t-s)\omega}}{2}f, s\right) \right), \\
		& c_s^{\dagger}(e^{it\omega}\hat{f},t)
			= -i\left( \dot{B}\left(\frac{e^{i(t-s)\omega}}{2\omega}f, s\right)
				+ iB\left(\frac{e^{i(t-s)\omega}}{2}f, s\right) \right).
	\end{align*}
Together with \eqref{F} we have (1).
Let us fix arbitrarily $\Psi, \Phi \in \mathcal{D}_B$
and define the function $\beta(s)$ by
$\beta(s) = (\Phi | c_s(h) \Psi)$.
Under the assumption of (2), we have
	\[ \frac{d}{ds}\beta(s)
		=  i\left(\Phi \Big|
\left(
B(\Delta g_s,s)-B(\partial_s^2g_s,s) \right) \Psi \right)= 0. \]
Hence, by the arbitrariness of $\Psi \in \mathcal{D}_B$, we obtain the desired results.
\qed	

By virtue of the above lemma, we introduce the definition of the positive (resp. negative) frequency part
of a given family in the class $\mathscr{D}(\mathcal{K})$.
\begin{definition}
{\bf (Positive frequency part and physical subspace)}\\
\label{dis4}
(1)
Let $\{B(\cdot,t)\}_{t\in\RR}\in \df$ and
\kak{suzuki2} be satisfied.
Then
we call $c(e^{-it\omega}\tilde{\hat{f}})$
(resp.\  $c^{\dagger}(e^{it\omega}\hat{f})$)
the positive (resp. negative) frequency part of $B(f,t)$
and denote it by
\eq{sa12}
B^{(+)}(f,t) := c(e^{-it\omega}\tilde{\hat{f}}),
		\quad
(resp. B^{(-)}(f,t) := c^{\dagger}(e^{it\omega}\hat{f})).
\en
(2)
Let
$B=\{B(\cdot,t)\}_{t\in\RR}\in \mathscr{D}(\mathcal{K})$ and $c_t(h)$ be defined by \kak{dis2}.
For each $t \in \mathbb{R}$,
we define the physical subspace $\V^t$ by
\eq{sa2}
\V^t := \{ \Psi \in \mathscr{D}_B | c_t(h)\Psi = 0,~ h \in \mathscr{S}(\mathbb{R}^3) \}.
\en
\end{definition}
\begin{remark}
In the abstract setting the physical  subspace $\V^t$
depends on time $t$.
The physical subspace associated with
free fields is, however, independent of $t$.
More precisely, assume that
$\{B(\cdot,t)\}_{t\in\RR}\in \df$ and \kak{suzuki2} is satisfied.
Then
 $\V =\V^t$ is independent of $t \in \mathbb{R}$.
 \end{remark}

\subsection{Physical subspace at time $t$} \label{subsec5.2}
We return to NRQED.
Applying the abstract theory explained in the previous section, we shall construct a physical subspace at time $t<\infty$ as the kernel of some operator.
   First we define an operator valued distribution. Let
\eqnn
\label{dis5}
&&		\B_P(f,t) :=  \partial^\mu \A_\mu(f,t,P),\\
&&\label{dis55}
\dot \B_P(f,t) :=  \partial^\mu \dot \A_\mu(f,t,P).
\ennn
More precisely
the right-hand side of \kak{dis5} and \kak{dis55}
 are abbreviations of
\eqnn
\label{defmu}
&&
\hspace{-1cm}
\partial^\mu \A_\mu(f,t,P)=
\partial_t
\A_0(f,t)+\A_1(\partial_{x^1} f,t,P)
+
\A_2(\partial_{x^2} f,t,P)
+
\A_3(\partial_{x^3} f,t,P),\\
&&\hspace{-1cm}
\label{defmuu}
\partial^\mu \dot \A_\mu(f,t,P)=
\partial_t  \dot \A_0(f,t)+
\dot \A_1(\partial_{x^1} f,t,P)
+\dot \A_2(\partial_{x^2} f,t,P)+\dot \A_3(\partial_{x^3} f,t,P).
\ennn
Then $\{\B_P(\cdot,t)\}_{t\in\RR}\in \mathscr {D}(\ffff)$.
Define the positive frequency part of $\B_P(h,t)$ as
	\[ c_{P,t}(h) := i\left(\dot{\B_P}(h_t,t)-\B_P(\dot h_t,t) \right), \]
where the functions $h_t$ and $\dot h_t$ are
defined as in \eqref{gt}, and
the physical subspace is defined as
\eq{ps}
 \V ^t_{P,{\rm phys}}
		:= \{ \Psi \in \ffff  | c_{P,t}(h) \Psi = 0,~ h \in \mathscr{S}(\mathbb{R}^3) \}.
\en
Of course, in general,
$\V ^t_{\rm phys}$ is not independent of
time $t$. To characterize $\V ^t_{P,{\rm phys}}$, we introduce  unitary operators.
Let $\Gamma([\gamma])$
be a unitary operator defined by the second quantization of the unitary operator
\eq{suzuki4}
[\gamma]:=\lk\begin{array}{cccc}
1&0&0&0 \\
0&1&0&0\\
0&0&1/\sqrt 2&1/\sqrt 2\\
0&0&1/\sqrt 2&-1/\sqrt2 \end{array} \rk:\oplus^4\LR\rightarrow \oplus^4\LR,
\en
Furthermore we
define the unitary operator $\W$ by
\eq{sa3}
\W  := \exp\left(
-\frac{e}{\sqrt{2}}
 \left(a^\ast\left(\frac{\hat{\varphi}}{\omega^{3/2}},3\right)
 -a\left(\frac{\hat{\varphi}}{\omega^{3/2}},3\right) \right)
\right)\Gamma([\gamma]).
\en
\begin{theorem}
\label{TheVt}
$\V_{P,{\rm phys}}^t$ is positive semi-definite and
	\eq{mainsuzuki}
 \V _{P,{\rm phys}}^t = e^{itH_{\pp,P}} e^{-it\hff}\W \ffff_{\pp}^{(0)},
 \en
 where
$\ffff_{\pp}^{(0)} = \ffff_{\pp}\otimes \{\alpha \Omega_0 | \alpha \in \mathbb{C} \}$.
\end{theorem}
\proof
We notice that
	\begin{equation}
	\label{c_0(g)}
		c_{P,0}(h) = a(\sqrt{\omega}h,3) - a(\sqrt{\omega}h,0)
			+ \frac{e}{\sqrt{2}}(\bar h, \vp/\omega),
	\end{equation}
and by the definition of
$c_{P,t}(h)$ we can observe that
\[ c_{P,t}(h) = \frac{i}{\sqrt{2}}e^{itH_{\pp,P}}
e^{-it\hff}c_0(h)
e^{it\hff}e^{-itH_{\pp,P}}. \]
Moreover, it follows directly that
$\W ^{-1}c_{P,0}(h)\W  = \sqrt{2}a(\sqrt{\omega}h,0)$, where  we have used
	\begin{align*}
		& \Gamma([\gamma]) a(f,j)\Gamma([\gamma])^{-1}= a(f,j),~~j=1,2, \\
	& \Gamma([\gamma]) a(f,3)\Gamma([\gamma])\f =\frac{1}{\sqrt{2}}[a(f,3)+a(f,0)], \\
		& \Gamma([\gamma]) a(f,0)\Gamma([\gamma])\f =\frac{1}{\sqrt{2}}[a(f,3)-a(f,0)].
	\end{align*}
Hence \kak{mainsuzuki} follows from
the equality $\{\Psi\in\ffff|c_{P,0}(h)\Psi=0\}=\W\ffff_{\pp}^{(0)}$.
Let $\Psi=e^{itH_{\pp,P}}e^{-it\hff}\W \Phi\in \V^t$, where $\Phi\in \ffff_{\pp}^{(0)}$.
Then
$(\Psi|\Psi)=(\Gamma([\gamma])\Phi, \eta\Gamma([\gamma])\Phi)=(\Phi,
-\Gamma([\gamma g \gamma])\Phi)\geq0$. Here
$$-[\gamma g \gamma]=
\lk\begin{array}{cccc}
1&0&0&0 \\
0&1&0&0\\
0&0&0&1\\
0&0&1&0\end{array} \rk:\oplus^4\LR\rightarrow \oplus^4\LR
$$
denotes the interchange between the $0$th and $3$rd components.
Then the theorem is complete.
\qed

$\W$
leaves the transversal part $\ffff_1\otimes \ffff_2$ invariant,
$\ffff_0$ and $\ffff_3$ are, however, mixed together by $\W$.
   Although the Hamiltonian $H_P =H_{\pp,P}\otimes 1+1\otimes H_0$ is subdivided into a scalar and a vector component , the physical subspace
is, however, of a more complicated form.

\subsection{Physical subspace at $t=\pm\infty$}
In this subsection we consider the physical subspace at $t=\pm \infty$. We have already presented the explicit form of the asymptotic field
$a_{P,\oi}(h,\mu)$, $\mu=0,1,2,3$,
and proven  its asymptotic completeness.
We can construct the free field in terms
of $a_{P,\oi}(h,\mu)$ and define the physical subspace independent of $t$.

Formally, we write
	\[ a_{P,\oi} ^\sharp(h, \mu) = \int
 h(k) a_{P,\oi} ^\sharp(k, \mu)dk. \]
We now define
 the smeared field
 $\A_{\mu}^\oi (f,t,P)$ in terms of $\ass_{P,\oi}$ by
	\begin{align}
		\A _j^\oi (f,t,P)
		& = \mmml\int
 dk\frac{e_j^l (k)}{\sqrt{2\omega(k)}}
			\left(a_{P,\oi} ^{\dagger}(k,l)\hat{f}(k)e^{i\omega(k)t}	
+a_{P,\oi} (k,l)\hat{f}(-k)e^{-i\omega(k)t}\right), \non \\
& \hspace{10cm}j=1,2,3,\\
		\A _0^\oi (f,t,P)
		& = \int \frac{dk}{\sqrt{2\omega(k)}}
			\left(a_{P,\oi} ^{\dagger}(k,0)\hat{f}(k)e^{i\omega(k)t}
				+a_{P,\oi} (k,0)\hat{f}(-k)e^{-i\omega(k)t}\right).
	\end{align}
Let us define the operator valued distribution
	\begin{align}
\label{freef}
		\B_{P,\oi}  (f,t) :=
 \partial^{\mu}\A^\oi _{\mu}(f,t,P),\quad \dot \B_{P,\oi}  (f,t) :=
 \partial^{\mu}\dot \A^\oi _{\mu}(f,t,P).
\end{align}
Here the right-hand side of \kak{freef} is understood as in \kak{defmu} and \kak{defmuu} with $\A_{\mu}$ replaced
by $\A_\mu^{\oi}$.
Then $\{\B_{P,\oi}  (\cdot,t)\}_{t\in\RR}\in \mathscr{D}(\ffff)$.
In addition, by the definition of $\A_\mu^{\oi}$ it is clear that
$\{\B_{P,\oi}  (\cdot,t)\}_{t\in\RR}\in \dff$.
From Lemma \ref{positivepart},
the positive frequency part
	\eq{suzukifind}
c_{P,\oi} (g) :=
i\left(\dot{\B}_{P,\oi}  (g_t,t)-\B_{P,\oi}  (\dot g_t,t) \right)
\en
is independent of $t$, and
 the physical subspace at time $t=\pm\infty$ is defined by
	\eq{ps2}
\V^\oi _{P,{\rm phys}}
		:= \{ \Psi \in \ffff  | c_{P,\oi} (h) \Psi = 0,~ h \in \mathscr{S}(\mathbb{R}^3) \}.
\en
Let
\eq{sa5}
\W_P:=U_P\W
\en
We can characterize the physical subspace $\V_P^\oi$ in the theorem below.
\begin{theorem}
\label{suzuki9}
Both $\V _{P,{\rm phys}}^{\rm in}$ and
$\V _{P,{\rm phys}}^{{\rm out}}$
are positive semi-definite and
\eqnn
\label{main2}
&&\V _{P,{\rm phys}}^{\rm in} = \W_P
\ffff_{\pp}^{(0)}, \\
\label{main3}
 &&\V _{P,{\rm phys}}^{{\rm out}} = S_P^{-1}\W_P  \ffff_{\pp}^{(0)}.
\ennn
\end{theorem}
\proof
Directly, we have
\eq{suzuki8}
c_{P,\oi}(h)= \frac{i}{\sqrt{2}}
\lk  a_{P,\oi} (\sqrt{\omega}h,3)-a_{P,\oi} (\sqrt{\omega}h,0)\rk . \en
Here we have used $\sum_{j=1,2}\mmml k_l e_l^j=0$.
In particular
$$U_P\f c_{P,{\rm in}}(h)U_P
=\frac{i}{\sqrt 2}\lk a(\sqrt\omega h,3)-a(\sqrt\omega h,0)\rk$$
follows. Then \kak{main2} follows.
 By $ a_{P,{\rm in}}^\sharp (h,\mu) S_P= S_P
 a_{P,{\rm out}}^\sharp(h,\mu)$,
 \kak{main3} also follows.
Finally the semidefinite property of
$\V_{P,{\rm phys}}^\oi$ can be obtained from the fact that $S_P$ is $\eta$ unitary and $[U_P,\eta]=0$. Then the proof is complete.
\qed

By Theorem \ref{suzuki9},
  we observe that
	$
 \V  _{P,{\rm phys}}^{\rm in} =
 S_P\V _{P,{\rm phys}}^{{\rm out}}$.
 We find, however, that $\V  _{P,{\rm phys}}^{\rm in}$ is {\it not} identical to
$\V  _{P,{\rm phys}}^{{\rm out}}$.
\bt{suzu34}
We have
$\V  _{P,{\rm phys}}^{\rm in} \not=
\V _{P,{\rm phys}}^{{\rm out}}$.
\et
\proof
Let $\Psi=\W_P \Phi\in \V_{P,{\rm phys}}^{\rm in}$ with $\Phi\in \ffff_{\pp}^{(0)}$.
Then we have
$$c_{P,{\rm out}}(h)\Psi=
\frac{i}{\sqrt 2}U_P
\lk
\mmml a(L^{l3}\sqrt\omega h,l)-a(\sqrt\omega h,0)\rk
\W\Phi.$$
One can easily find some vector $\Phi\in \ffff_{\pp}^{(0)}$ such that the right-hand side above does not vanish.
Then the theorem follows.
\qed


\section{Physical Hamiltonian}
\subsection{Physical Hilbert space and physical scattering operator}
We defined $\V_{P,{\rm phys}}^\ex$ in the previous section and it includes  the null space $\V_{P,{\rm null}}^\ex$ with respect to $(\cdot|\cdot)$. We want to define the physical Hilbert space by $\V_{P,{\rm phys}}^\ex$ divided by the null space and
a self-adjoint physical Hamiltonian on it.

We first of all characterize the null space of $\V_{P, {\rm phys}}^\ex$.
Let $\ffff_\T=\ffff_1 \otimes \ffff_2$ and
\eq{LLTT}
\ffff^{(0)}_\T = \ffff_\T\otimes\{\alpha \Omega_3 | \alpha \in \CC \}\otimes\{\alpha \Omega_0 | \alpha \in \CC\}.
\en
Then $\ffff^{(0)}_{\pp}$ can be decomposed as
 $\ffff^{(0)}_{\pp}=
     \ffff^{(0)}_\T \oplus \lk  \ffff^{(0)\perp}_\T \cap \ffff^{(0)}_{\pp}\rk$.
Let
\eqnn
&&		\V_P^{\rm in} := \W_P\ffff^{(0)}_\T, \quad
		\V_{P,{\rm null}}  ^{\rm in} :=
\W_P\ \lk  \ffff^{(0)\perp}_\T \cap \ffff^{(0)}_{\pp}\rk,\\
&&
\V_P^{\rm out} := S_P^{-1}\V_P^{\rm in},\quad
\V_{P, {\rm null}}^{\rm out}:=S_P^{-1}\V_{P,{\rm null}}^{\rm in}.
\ennn
Then
 $\V_{P,{\rm phys}}^\oi$ is also decomposed as \eq{hi1}
\V_{P,{\rm phys}}^\oi= \V_P^\oi\oplus \V_{P,{\rm null}}^\oi.\en
Here  $\V_P^{\ex}$ is closed, positive definite and satisfies
$(\Psi_1|\Psi'_1) = (\Psi_1,\Psi'_1)$ for
$\Psi_1, \Psi'_1 \in \V_P^{\ex}$,
and
 $\V_{P,{\rm null}}  ^{\ex}$ is closed, neutral and
	\[ \V_{P,{\rm null}}  ^{{\ex}} = \{ \Psi_0 \in \V_{P,{\rm phys}}^{\ex} | ( \Psi_0 | \Psi_0)=0 \}. \]
\begin{definition}
For subspaces $Y$, $Z$ and $X$
in $\ffff$, we use the notation
$$X =Y \+ Z$$ if and only if
    (1) for all $x \in X$,
    there exist unique vectors
    $y \in Y$ and $z\in Z$
such that $x = y+z$,
 (2)
 $(y|z) =0$ holds for all $y\in Y$ and $z\in Z$.
\end{definition}
\begin{lemma}
It follows that
$\V_{P,{\rm phys}}^{\ex} = \V_P^{\ex} \+ \V_{P,{\rm null}}  ^{\ex}$.
\end{lemma}
\proof
Since $\V_{P,{\rm phys}}^{\ex}$ is positive semi-definite with respect to the metric $(\cdot|\cdot)$,
we have, by the Schwartz inequality,
$|(\Psi_1| \Psi_0)|^2 \leq (\Psi_1| \Psi_1)(\Psi_0| \Psi_0)=0$
for $\Psi_1 \in \V_P^{\ex}$ and $\Psi_0 \in \V_{P,{\rm null}}  ^{\ex}$.
Hence $(\Psi_1| \Psi_0)=0$.
Then
the lemma follows.
\qed

We define the physical Hilbert space
in terms of the quotient Hilbert space
\eq{phi}
 \hhh_{P,{\rm phys}}^{\ex} := \V^{\ex}_{P,{\rm phys}}/\V_{P,{\rm null}}  ^{\ex}.
 \en
 We denote by $[\Psi]_{\ex}$ the element of $\hhh_{P,{\rm phys}}^{\ex}$
associated with $\Psi \in
\V_{P,{\rm phys}}^{\ex}$ and the induced
 scalar  product on
 $\hhh_{P,{\rm phys}}^{\ex}$ is denoted by $(\cdot, \cdot)_{\ex}$, i.e.,
 $(\E \Psi, \E \Phi)_{\ex}=(\Psi,\Phi)$.
Furthermore let $\pi_\ex:\V_p^\ex\rightarrow \hhh_{P,{\rm phys}}^\ex$ be the natural onto map defined by $\pi_\ex(\Phi):=\E \Phi$.
Thus $\pi_\ex$ is an isometry and so is
a unitary operator between
$\V_p^\ex$ and
$\hhh_{P,{\rm phys}}^\ex$.



We have already defined the scattering operator $S_P$.
This operator maps the null space $\V_{P, {\rm null}}^{\rm out}$ into
the null space
$\V_{P, {\rm null}}^{\rm in}$. Now
we can define
the physical scattering operator.
\begin{definition}
The physical scattering operator
$S_{P, {\rm phys}}:\hhh_{P,{\rm phys}}^{\rm out} \longrightarrow \hhh_{P,{\rm phys}}^{\rm in}$ is defined by
\eq{scatter}
 S_{P, {\rm phys}}[\Psi]_{\rm out} := [S_P\Psi]_{\rm in}.
\en
\end{definition}
\begin{theorem}{\bf (Physical scattering operator)}
\label{suzu2}
$S_{P, {\rm phys}}$ is unitary.
\end{theorem}
\proof
Since $S_P$ is a unitary operator from
$\V_P^{\rm out}$ to $\V_P^{\rm in}$, the theorem follows.
\qed

\subsection{Physical Hamiltonian}
In the previous section we defined the physical Hilbert space. Next
we  define the physical Hamiltonian $H^{\ex}_{P,{\rm phys}}$
on $\hhh_{P,{\rm phys}}^{\ex}$
and prove its self-adjointness.

We define
	\[ P^{\rm in}:= \W_P  P_{\pp} \W_P ^{-1}, \quad
		P^{\rm out}:=S^{-1} \W_P  P_{\pp} \W_P ^{-1}S. \]
Here $P_{\pp}= 1 \otimes 1 \otimes 1 \otimes P_{\Omega_0}$
is the orthogonal projection onto $\ffff_{\pp}^{(0)}$,
where $P_{\Omega_0}$ is the orthogonal projection onto $\{\alpha\Omega_0 | \alpha \in \CC \}$.
Then
   $P^{\ex}$ is the orthogonal projection onto $\V_{P,{\rm phys}}^{\ex}$.
We have to say something about relationships between the domain of $H_P $ and $\V_{P, {\rm phys}}^\oi$.
\begin{lemma} \label{invP}
\begin{enumerate}
\item[(1)] $P^{\ex}$ leaves $D(H_P )$ invariant, i.e., $P^\ex D(H_P ) \subset D(H_P )$.
\item[(2)] $H_P $ leaves $\V_{P,{\rm phys}}^\ex$ invariant,
i.e., $H_P (D(H_P ) \cap \V_{P,{\rm phys}}^\ex) \subset \V_{P,{\rm phys}}^\ex$.
\end{enumerate}
\end{lemma}
\proof
Let us define the operator
$$
\tilde H_P := H_{\rm f}-a^\ast (\hat{\varphi}/\sqrt{\omega},3)
		 -a(\tilde{\hat{\varphi}}/\sqrt{\omega},0)+
E_{\pp}(P)+E_0$$
with domain $D(\tilde H_P)=D(\hf)$.
Since
$-a^\ast (\hat{\varphi}/\sqrt{\omega},3)
		 -a(\tilde{\hat{\varphi}}/\sqrt{\omega},0)
$ is infinitesimally small with respect to $\hf$,
we have $\|\hf\Psi\|\leq C(\|\tilde H_P\Psi\|+\|\Psi\|)$ for some constant $C$. Then
$\tilde H_P$ is closed. Furthermore by
 $\|\tilde H_P\Psi\|\leq c(\|\hf\Psi\|+\|\Psi\|)$ for some constant $c$,
  an arbitrary core of $\hf$ is also a core of $\tilde H_P$.
Thus both
$H_P $ and $\tilde H_P$
are closed and have the same domain $D(H_P )=D(H_{\rm f})=D(\tilde H_P)$;
moreover
have the common core
	\[ \fffff(\omega):={\rm L.H.}
\lkk
\left.
\PI \add(f_i,\mu_i)
 \Omega,\Omega \right|
 f_i\in D(\omega),\mu_i=0,1,2,3, i=1,...,n,n\geq 1\rkk. \]
It is immediate that
$\W_P\f H_P \W_P = \tilde H_P$  on
the common core $\fffff(\omega)$.
Then
$\W_P D(H_P ) \subset D(H_P )$ and
we have the operator equation
$\W_P\f H_P \W_P = \tilde H_P$.
Since, by
$P_{\pp}H_{\rm f} \subset H_{\rm f}P_{\pp}$,
$P_{\pp}$ leaves $D(H_P )$ invariant,
we have
	\begin{align*}
		& P^{\rm in}D(H_P ) \subset \W_P  P_{\pp}D(H_P ) \subset \W_P  D(H_P ) \subset D(H_P ), \\
		& P^{\rm out}D(H_P ) \subset S_P^{-1}\W_P  P_{\pp}D(H_P ) \subset S_P^{-1}\W_P  D(H_P ) \subset D(H_P ),
	\end{align*}
where we have used
the intertwining property
$S_PH_P =H_P S_P^{-1}$.
Thus the first half of the lemma is proven.
For $\Psi \in D(H_P ) \cap \V_{P,{\rm phys}}^{\rm in}$, we have
	\begin{align}
		H_P \Psi & = H_P P^{\rm in}\Psi = \W_P \tilde H_PP_{\pp}\W_P ^{-1}\Psi \notag  \\
	\label{notredin}
		& = \W_P  P_{\pp}\lk H_{\rm f}-a^\ast (\hat{\varphi}/\sqrt{\omega},3)
				+E_{\pp}(P)+E_0 \rk \W_P ^{-1}\Psi \in \V_{P,{\rm phys}}^{\rm in}.
	\end{align}
Then
for $\Psi \in D(H_P ) \cap \V_{P,{\rm phys}}^{\rm out}$,
	\begin{align}
		H_P \Psi & = S_P^{-1}H_P P^{\rm out}\Psi = \W_P \tilde H_PP_{\pp}\W_P ^{-1}S_P\Psi \notag  \\
	\label{notredout}
		& = S_P^{-1}\W_P  P_{\pp}\lk H_{\rm f}-a^\ast (\hat{\varphi}/\sqrt{\omega},3)
				+E_{\pp}(P)+E_0 \rk \W_P ^{-1}S_P\Psi \in \V_{P,{\rm phys}}^{\rm out}.
	\end{align}
Hence the proof is complete.
\qed


Let $K_P^\ex$ be the restriction of $H_P$ to $D(H_P)\cap \V_{P, {\rm phys}}^\ex$:
\eq{rest}
K_P^\ex:=H_P \lceil_{D(H_P )\cap \V_{P,{\rm phys}}^{\ex}}.
\en
By Lemma \ref{invP}
$K_P^\ex$ is a densely defined closed operator on $\V_{P,{\rm phys}}^{\ex}$. Note that $\V_{P,{\rm phys}}^{\ex}$ is closed.
In order to study $K_P^\ex$ we introduce
the operator
	\[ \hat H_P :=H_{\rm f}-a^\ast (\hat{\varphi}/\sqrt{\omega},3) + E_{\pp}(P) + E_0 \]
	with domain $D(\hat H_P)=D(\hf)$.
In a similar way as in the proof that $\tilde{H}_p$ is closed,
one can determine that $\hat H_P$ is closed.
By \kak{notredin} and \kak{notredout},
we have, for all $\Psi \in
D(K_P^\ex )$,
	\begin{align}
	\label{redHin}
		& K_P^{\rm in}\Psi = \W_P  \hat H_P  \W_P ^{-1}\Psi, \\
	\label{redHout}
		& K_P^{\rm out}\Psi = S_P^{-1}\W_P  \hat H_P  \W_P ^{-1}S_P\Psi.
	\end{align}
\begin{lemma}\label{redinv}
\begin{enumerate}
\item[(1)] $K_P^{\rm in}$ is reduced by $\V_{P,{\rm phys}}^{\rm in}$,
i.e., $P^{\rm in} K_P ^{\rm in} \subset K_P^{\rm in}P^{\rm in}$.
\item[(2)] $K_P^\ex $ leaves $\V_{P,{\rm null}}  ^{\ex}$ invariant,
i.e., $K_P^\ex (D(K_P^\ex ) \cap \V_{P,{\rm null}}  ^{\ex})
\subset \V_{P,{\rm null}}  ^{\ex}$.
\end{enumerate}
\end{lemma}
\proof
(1) follows from \eqref{redHin}, \eqref{redHout} and the fact
that $\hat H_P $ is reduced by $\ffff_{\pp}^{(0)}$, i.e., $P_{\pp}\hat H_P  \subset \hat H_P  P_{\pp}$.
Let $\Psi_0 \in D(K_P^{\rm in})
\cap \V_{P,{\rm null}}  ^{\rm in}$
and set
$\Phi_0=\W_P ^{-1}\Psi_0 \in D(\hat H_P )$.
Then $\Phi_0 \in  \ffff^{(0)\perp}_\T \cap \ffff_\pp^{(0)}$.
Since $\hat H_P \Phi_0 \in  \ffff^{(0)\perp}_\T \cap \ffff_\pp^{(0)}$,
we have
$K_P^{\rm in}\Psi_0
			= \W_P  \hat H_P \Phi_0 \in \V_{P,{\rm null}}  ^{\rm in}$.
Thus $K_P^{\rm in}$ leaves $\V_{P,{\rm null}}  ^{\rm in}$ invariant.
Similarly, one can prove that
$K_P^{\rm out}$ leaves $\V_{P,{\rm null}}  ^{\rm out}$ invariant.
\qed
We denote by $\rho(X)$
the resolvent set of a linear operator $X$.
By
    \eqref{redHin} and \eqref{redHout}
it follows that $\rho(\hat H_P ) \subset \rho(K_P^\ex )$
and
for  $z \in \rho(\hat H_P )$,
\begin{eqnarray}
\label{resol}
		& (K_P^{\rm in}-z)^{-1}
			= \W_P  (\hat H_P -z)^{-1}
\W_P\f, \\
\label{resol2}
		& (K_P^{\rm out}-z)^{-1}
			= S_P^{-1}\W_P  (\hat H_P -z)^{-1}\W_P\f S_P.
\end{eqnarray}
Let us now define the physical Hamiltonian $H_{P,{\rm phys}}^{\ex}$
on the physical Hilbert space $\hhh_{P,{\rm phys}}^{\ex}$.
In order to define the domain $D(H^{\ex}_{\rm phys})$ consistently,
we first consider the resolvent of $K_P^\ex $.

Let
$$
 \mathcal{R} = \lkk
 z \in \rho\left(
 H_{\rm f}+E_\pp(P)+E_0\right)
		\Big| 2\epsilon + \frac{{\|\hat{\varphi}/\omega\|^2}/({2\epsilon}) + {\|\hat{\varphi}/\sqrt{\omega}\|}/{\sqrt{2}}      }{|E_\pp(P)+E_0-z|} <1~ \mbox{for some $\epsilon >0$} \rkk.
$$
Since  $\|a^\ast (\hat{\varphi}/\sqrt{\omega},3)(H_{\rm f} +E_\pp(P)+E_0-z)^{-1}\| <1$ for $z\in\mathcal {R}$,
   for all $z \in \mathcal{R}$,
   the Neumann expansion is valid and
 	\begin{equation}
	\label{resolexp}
		(\hat H_P  -z)^{-1} = \sum_{n=0}^{\infty}(H_{\rm f}+
E_\pp(P)+E_0-z)^{-1}
		\lk
a^\ast (\hat{\varphi}/\sqrt{\omega},3)(H_{\rm f} +E_\pp(P)+E_0-z)^{-1}\rk
^n.
	\end{equation}

Let us fix $z \in \mathcal{R}$.
Then, by \kak{resol} and \kak{resol2}, $z \in \rho(K_P^\ex )$ and
the resolvent
\eq{res3}
 R_P^\ex (z):= (K_P^\ex -z)^{-1}
 \en
 is  bijective on $\V_{P,{\rm phys}}^{\ex}$.
\begin{lemma}\label{voinv}
$R_P^\ex (z)$ is reduced by $\V_{P,{\rm phys}}^{\ex}$ and leaves $\V_{P,{\rm null}}  ^{\ex}$ invariant,
i.e., $P^{\ex}R_P^\ex (z) = R_P^\ex (z)P^\ex$ and $R_P^\ex (z) \V_{P,{\rm null}}  ^{\ex} \subset \V_{P,{\rm null}}  ^{\ex}$.
\end{lemma}
\proof
The first half of this lemma has already
been proven via Lemma \ref{redinv} (1).
We prove the second half.
Let $\Psi_0 \in \V_{P,{\rm null}}  ^{\rm in}$
and set
$\Phi_0 = \W_P ^{-1}\Psi_0$.
Then
     $\Phi_0 \in  \ffff^{(0)\perp}_\T \cap \ffff_\pp^{(0)}$.
By \kak{resol}, \kak{resol2} and \eqref{resolexp},
we observe that
\eq{wakaran}
R_P^{\rm in}(z) \Psi_0 =
\W_P
(\hat H_P -z )^{-1}\Phi_0 \in
\W_P(
  \ffff^{(0)\perp}_\T \cap \ffff_\pp^{(0)})=\V_{P, {\rm null}}^{\rm in}.
\en
Thus  $R_P^{\rm in}(z)$ leaves $\V_{P, {\rm null}}^{\rm in}$ invariant.
Similarly one can prove that $R_P^{\rm out}(z)$ also leaves $\V_{P,{\rm null}}  ^{\rm out}$ invariant.
\qed
Since $R_P^{\rm in}(z)$ leaves the null space invariant,
the following operator, $[R_P^\ex (z)]_{\ex}$, on $\hhh_{P,{\rm phys}}^{\ex}$ is well-defined:
	$$ [R_P^\ex (z)]_{\ex}[\Psi]_{\ex} := [R_P^\ex (z)\Psi]_{\ex}.$$
It is clear
that $[R_P^\ex (z)]_{\ex}$ is bounded
and $\|[R_P^\ex (z)]_{\ex}\|_\ex  \leq \|R_P^\ex (z)\|$ holds.
\begin{lemma}\label{inj}
$[R_P^\ex (z)]_{\ex}$ is injective and $[R_P^\ex (z)]_{\ex}^{-1}$ is closed.
\end{lemma}
\proof
By the boundedness of $[R_P^\ex (z)]_{\ex}$,
$[R_P^\ex (z)]_{\ex}^{-1}$ is closed if $[R_P^\ex (z)]_{\ex}$ is injective.
Let $[R_P^\ex (z)]_{\ex}[\Psi]_{\ex}=0$.
Then  
$R_P^\ex (z)\Psi \in \V_{P,{\rm null}}  $.
It follows from Lemma \ref{redinv} (2) that
$\Psi = (K_P^\ex -z )R_P^\ex (z)\Psi \in \V_{P,{\rm null}}  ^{\ex}$.
Thus $[\Psi]_{\ex}=0$ and $[R_P^\ex (z)]_{\ex}$ is injective.
\qed
\begin{definition}
We define the physical Hamiltonian
$H_{P,{\rm phys}}^{\ex}$ on $\hhh_{P,{\rm phys}}^{\ex}$ by
\eq{ph}
 H_{P,{\rm phys}}^{\ex} := z  + [R_P^\ex (z)]_\ex^{-1}. \en
\end{definition}
By Lemma \ref{inj}, $H_{P,{\rm phys}}^{\ex}$ is closed. We further prove
that  $H_{P,{\rm phys}}^{\ex}$ is independent of  $z  \in \mathcal{R}$
and that the domain of $H_{P,{\rm phys}}^{\ex}$ is dense in
$\hhh_{P,{\rm phys}}^{\ex}$.
We define $P^{\ex}$ by
	\begin{equation}
	\label{P1ex}
		P_1^{\rm in}:= \W_P  P_{\T} \W_P ^{-1}, \quad
		P_1^{\rm out}:=S_P^{-1} \W_P  P_{\T} \W_P ^{-1}S_P.
	\end{equation}
Here $P_{\T}=1\otimes 1\otimes P_{\Omega_3} \otimes P_{\Omega_0}$.
Then $P_1^{\ex}$
is the orthogonal projection onto $\V_P^{\ex}$.
We define a linear operator $\J $ on
$\hhh_{P,{\rm phys}}^{\ex}$ by
\begin{eqnarray}
\J [\Psi]_{\ex}
&=& [K_P^\ex P_1^{\ex}\Psi]_{\ex},\\
D(\J )
&=& \left\{ [\Psi]_{\ex} \in \hhh_{P,{\rm phys}}^{\ex}
				| P_1^{\ex}\Psi \in D(K_P^\ex ) \right \}.
\end{eqnarray}
Note that  the domain $D(\J )$
is independent of  the representative.
Indeed if $[\Psi]_{\ex}=[\Psi']_{\ex}$, then $P_1^{\ex}\Psi=P_1^{\ex}\Psi'$
because $\Psi-\Psi' \in \V_{P,{\rm null}}  ^{\ex}$ and  $P_1^{\ex}\V_{P,{\rm null}}  ^{\ex}=\{0\}$.
Thus the operator $\J $ is well-defined.
Moreover
since $P_1^{\ex}$ leaves $D(K_P^\ex )$ invariant,
$D(\J )$ is dense in $\hhh_{P,{\rm phys}}^{\ex}$.
\bl{equivH}
It follows that
\eq{J}
H_{P,{\rm phys}}^{\ex} = \J .
\en
In particular, $H_{P,{\rm phys}}^{\ex}$ is a densely defined closed operator
and  independent of $z  \in \mathcal{R}$.
\el
\proof
Let $[\Psi]_{\ex} \in D(\J )$.
Then  $P_1^{\ex}\Psi \in D(K_P^\ex )$ and set
$\Phi_0 := (K_P^\ex -z )P_1^{\ex}\Psi$.
We observe that
	\[ [\Psi]_{\ex}
=[P_1^{\ex}\Psi]_{\ex}
		= [R_P^\ex (z)]_{\ex}[\Phi_0]_{\ex} \in D([R_P^\ex (z)]_{\ex}^{-1}) \]
and hence  $D(\J)\subset D(H_{P,{\rm phys}}^{\ex})$.
We show the inverse inclusion.
Let
$[\Psi]_{\ex} \in  D(H_{P,{\rm phys}}^{\ex})$.
Then there exists a vector $[\Phi]_\ex \in \hhh_{P,{\rm phys}}^\ex$ such that
$[\Psi]_{\ex}=[R_P^\ex (z)]_{\ex}[\Phi]_\ex=[R_P^\ex (z)\Phi]_\ex$.
Since
$P_1^{\ex}$ leaves $D(K_P^\ex )$ invariant,
we have $P_1^\ex\Psi=P_1^{\ex} R_P^\ex (z)\Phi \in D(K_P^\ex )$.
Then
$D(\J)\supset D(H_{P,{\rm phys}}^{\ex})$ follows.
Thus
$$D(H_{P,{\rm phys}}^{\ex}) = D(\J ).$$
For all $[\Psi]_\ex \in D(H_{P,{\rm phys}}^{\ex})$,
we see that
(1) there exists $[\Phi]_\ex \in \hhh_{P,{\rm phys}}^{\ex}$ such that
$[\Psi]_{\ex}=[R_P^\ex (z)]_{\ex}[\Phi]_\ex$ and
(2) $P_1^{\ex}\Psi \in D(K_P^\ex )$.
We have
	\begin{align*}
		& H_{P,{\rm phys}}^{\ex}[\Psi]_\ex -  \J [\Psi]_\ex
		= [z \Psi +\Phi- K_P^\ex P_1^{\ex}\Psi]_\ex.
	\end{align*}
We need only prove that $z \Psi +\Phi- K_P^\ex P_1^{\ex}\Psi \in \V_{P,{\rm null}}  ^\ex$.
Together, (1) and (2) imply
that
  $P_1^{\ex}\Psi-R_P^\ex (z)\Phi \in \V_{P,{\rm null}}  ^\ex \cap D(K_P^\ex )$,
which, together with Lemma \ref{redinv}, yields
$$z \Psi +\Phi- K_P^\ex P_1^{\ex}\Psi
= z (1-P_1^{\ex})\Psi-(K_P^\ex -z )(P_1^{\ex}\Psi-R_P^\ex (z)\Phi) \in \V_{P,{\rm null}}  ^\ex.$$
Thus  the lemma follows.
\qed

Now we are in a position to state the main theorem in this section:
\begin{theorem}
\label{suzu1}
$H_{\rm phys}^\ex$ is self-adjoint and has a unique ground state
with energy $E_{\pp}(P) +E_0$.
\end{theorem}
\proof
For all $[\Psi]_{\rm ex} \in D(\J )$,
we have
$
\J [\Psi]_\ex
		=\pi_\ex P_1^{\ex} K_P^\ex P_1^{\ex}\pi_\ex\f [\Psi]_\ex.
$
This equality implies that
$$\J  \subset \pi_{\ex}P_1^{\ex}H_{P,{\rm phys}}^{\ex}P_1^{\ex}\pi_{\ex}^{-1}.$$
Conversely if $[\Psi]_{\rm ex} \in D(\pi_{\ex}P_1^{\ex}K_P^\ex P_1^{\ex}\pi_{\ex}^{-1})$,
then $P_1^{\ex}\Psi=P_1^{\ex}\pi_{\ex}^{-1}[\Psi]_\ex \in D(K_P^\ex )$
and henceforth
$D(\pi_{\ex}P_1^{\ex}H_{P,{\rm phys}}^{\ex}P_1^{\ex}\pi_{\ex}^{-1})\subset D(\J)$.
Thus
we have obtained the result that
\eq{first}
\J=\pi_{\ex}P_1^{\ex}K_P^\ex  P_1^{\ex}
\pi _{\ex}^{-1}.
\en
Combining Lemma  \ref{equivH} and \kak{first}, we establish that
$$H_{P,{\rm phys}}^{\ex}=\pi_\ex P_1^{\ex} K_P^\ex P_1^{\ex}\pi^{-1}_\ex.$$ Then $H_{P,{\rm phys}}^{\ex}$ is self-adjoint
if and only if
$P_1^{\ex}K_P^\ex P_1^{\ex}$ is self-adjoint.
By \eqref{redHin}, \eqref{redHout} and \eqref{P1ex}, we have
\eq{ww1}
		P_1^{\rm in}K_P^{\rm in} P_1^{\rm in}
			 = \W_P  P_\T \hat H_P P_\T \W_P ^{-1}
 = U_P P_\T(H_{\rm f}^{\rm T} +E_{\pp}(P) +E_0)P_\T U_P^{-1}
\en
and, by the intertwining property,
\begin{eqnarray}
		P_1^{\rm out}K_P^{\rm out} P_1^{\rm out}
			& =& S_P^{-1}\W_P  P_\T \hat H_P P_\T \W_P ^{-1}S_P \non\\
		& =& S_P^{-1}U_P P_\T(H_{\rm f}^\T +E_{\pp}(P) +E_0)P_\T U_P^{-1}S_P,
\end{eqnarray}
where $H_{\rm f}^{\T}=\sum_{j=1,2}\int\omega(k)a^\ast (k)a(k)dk$.
The above equations imply  that $ P_1^{\ex} K_P^\ex P_1^{\ex}$
is self-adjoint and hence we have  the desired properties.
\qed

{\bf Acknowledgements}

We thank A. Arai for useful discussions.
F. H. thanks JSPS for the award of a Grant-in-Aid for Science Research (B) Number 20340032. 
A. S. thanks JSPS for financial support.
We also thank Support Program for Improving Graduate School Education for financial support.

\end{document}